\documentclass[10pt,journal,compsoc]{IEEEtran}

\ifCLASSOPTIONcompsoc
  \usepackage[nocompress]{cite}
\else
  \usepackage{cite}
\fi
\ifCLASSINFOpdf
\else
\fi

\usepackage{float}
\usepackage{cite}
\usepackage{color}
\usepackage{pifont}

\usepackage[justification=centering]{caption}
\usepackage{subfigure}
\usepackage{multirow, multicol}
\usepackage{booktabs, ctable}
\usepackage{algorithm}
\usepackage{algorithmicx}
\usepackage[noend]{algpseudocode}
\usepackage{comment}
\usepackage{enumitem}
\usepackage{fancybox}

\usepackage{amsmath,amssymb,amsfonts}
\usepackage{graphicx}
\usepackage{textcomp}
\usepackage{xcolor}
\usepackage{graphics}
\usepackage{epsfig}
\usepackage{color}
\usepackage{bm}
\renewcommand{\algorithmicrequire}{\textbf{Data Owner:}} 
\newcommand{\cmark}{{\color{black}\ding{51}}}%
\newcommand{\xmark}{{\color{black}\ding{55}}}%
\newcommand{\pseudo}{{\mathcal{F}}}
\newcommand{\hash}{{\mathcal{H}}}
\newcommand{\key}{{\mathcal{K}}}

\newcommand{\enc}{{\ensuremath{\sf Enc}}}
\newcommand{\dec}{{\ensuremath{\sf Dec}}}

\newtheorem{theorem}{Theorem}

\newtheorem{definition}{Definition}
\newtheorem{proof}{Proof}

\def\BibTeX{{\rm B\kern-.05em{\sc i\kern-.025em b}\kern-.08em
    T\kern-.1667em\lower.7ex\hbox{E}\kern-.125emX}}
\begin{document}

\title{Non-interactive Multi-client Searchable Symmetric Encryption with Small Client Storage}

\author{Hanqi Zhang, Chang Xu, Rongxing Lu, ~\IEEEmembership{Fellow,~IEEE}, Liehuang Zhu, Chuan Zhang, \\ Yunguo Guan
  \IEEEcompsocitemizethanks{
    \IEEEcompsocthanksitem Hanqi Zhang, Chang Xu, Liehuang Zhu, and Chuan Zhang are with Beijing Institute of Technology.\protect\\
    E-mail: \{ zhanghanqi, xuchang, liehuangz, and zhangchuan\}@bit.edu.cn
    \IEEEcompsocthanksitem Rongxing Lu and Yunguo Guan are with Faculty of Computer Science, University of New Brunswick, Fredericton, Canada. \protect\\
    Email: \{RLU1, yguan4\}@unb.ca.
  }
}

\markboth{Journal of \LaTeX\ Class Files,~Vol.~14, No.~8, August~2015}%
{Shell \MakeLowercase{\textit{et al.}}: Bare Demo of IEEEtran.cls for Computer Society Journals}

\IEEEtitleabstractindextext{%
\begin{abstract}
Considerable attention has been paid to dynamic searchable symmetric encryption (DSSE) which allows users to search on dynamically updated encrypted databases. To improve the performance of real-world applications, recent non-interactive multi-client DSSE schemes are targeted at avoiding per-query interaction between data owners and data users. However, existing non-interactive multi-client DSSE schemes do not consider forward privacy or backward privacy, making them exposed to leakage abuse attacks. Besides, most existing DSSE schemes with forward and backward privacy rely on keeping a keyword operation counter or an inverted index, resulting in a heavy storage burden on the data owner side. To address these issues, we propose a non-interactive multi-client DSSE scheme with small client storage, and our proposed scheme can provide both  forward privacy and backward privacy. Specifically, we first design a lightweight storage chain structure that binds all keywords to a single state to reduce the storage cost. Then, we present a Hidden Key technique, which preserves non-interactive forward privacy through time range queries, ensuring that data with newer timestamps cannot match earlier time ranges. We conduct extensive experiments to validate our methods, which demonstrate computational efficiency. Moreover, security analysis proves the privacy-preserving property of our methods.

\end{abstract}

\begin{IEEEkeywords}
Searchable Encryption, Non-interaction, Multi-client, Time Range Query
\end{IEEEkeywords}}

\maketitle

\IEEEdisplaynontitleabstractindextext

\IEEEpeerreviewmaketitle

\IEEEraisesectionheading{\section{Introduction}\label{sec:introduction}}
\IEEEPARstart{T}he ubiquitous availability of cloud computing services has enabled more and more individuals and companies (regarded as data owners) to save data in cloud storage servers.
Data owners typically encrypt data before uploading it to the cloud because it is completely trusted. Although encryption provides security at a high level, it restricts data utility such as search and calculation. Therefore, Song et al. \cite{song2000practical} introduce searchable symmetric encryption (SSE) to realize keyword searches over the encrypted database.

Existing SSE schemes can be classified based on the number of clients, i.e., single-client and multi-client. 
Existing schemes typically \cite{bost2016ovarphiovarsigma,sun2021practical,bost2017forward,ghareh2018new} focus on the single-client architecture.
In the single-client architecture, the data owner outsources its encrypted data to the server, and then sends search queries.
Compared to the single-client framework, the multi-client framework is more practical in many scenarios, including medical data sharing \cite{XU2020394}, task recommendations \cite{8253516}, geographic location queries \cite{9155505, 9564032}, etc.  The reason is that the multi-client framework can support any number of clients conducting queries over the encrypted database. Therefore, we employ the multi-client architecture as the backbone and propose a multi-client SSE scheme in this paper.

Multi-client SSE schemes can be categorized into two types: interactive multi-client SSE schemes \cite{jarecki2013outsourced, faber2015rich, zuo2019dynamic} and non-interactive multi-client SSE schemes \cite{10.1007/978-3-319-45744-4_8, sun2020non, sun2018dynamic}.
For interactive multi-client SSE schemes, the data owner has to be online. The reason is that the data user is required to keep interacting with the data owner for every query to ask for the necessary search information.
Therefore, it is critical to construct  non-interactive SSE schemes that can provide high communication efficiency. Sun et al.\cite{10.1007/978-3-319-45744-4_8} first propose the concept of non-interaction and develop the first multi-client searchable encryption protocol to avoid per-query interaction between the data owner and clients. Furthermore, they design a new multi-client searchable encryption mechanism \cite{sun2020non} that achieves enhanced security, i.e.,  untrusted clients are considered.
However, the two schemes \cite{10.1007/978-3-319-45744-4_8, sun2020non} only focus on the static database. For static searchable encryption, scalability is not provided efficiently, as it is achieved  by rebuilding indices or other expensive techniques \cite{song2000practical,9155505,8253516}.
Sun et al.\cite{sun2018dynamic} presents a boolean searchable symmetric
encryption scheme in a dynamic and non-interactive manner, which supports arbitrary boolean queries in the multi-client scenario. Unfortunately, it cannot achieve forward and backward privacy.
 Forward privacy is an important property to prevent the cloud server from linking newly added files with past searches.
Backward privacy ensures that search queries do not leak the information whether the deleted documents contain the searched keywords.

\begin{table*}[]
  \centering
 \caption{Comparison of existing SSE schemes.
    $K$ represents total number of $[document, keyword]$ pairs in the database. $|W|$ and $|D|$ are the number of all keywords and the total number of documents in the database, respectively.
    $a_w$ is the number of updates for keyword $w$ and $n_w$ is the number of documents containing $w$.
    $O$ notation hides polylogarithmic factors.
    'RT' means the number of round trips of the retrieved results.
    'BP' represents whether the scheme achieves backward privacy.
    'Multi-client' denotes whether the SSE scheme supports the multi-client setting. 'NOI' denotes whether the SSE scheme avoids per-query interaction.}
\scalebox{0.98}{
    \begin{tabular}{|c|l|l|l|l|l|l|l|l|l|}
      \hline
      \multirow{2}{*}{Scheme}   & \multicolumn{2}{c|}{Computational cost} & \multicolumn{3}{c|}{Communication overhead} & \multicolumn{1}{c|}{Client}  &  \multicolumn{1}{c|}{\multirow{2}{*}{BP}}& \multicolumn{1}{c|}{Multi-} &\multicolumn{1}{c|}{\multirow{2}{*}{NOI}}  \\ \cline{2-6}
            & \multicolumn{1}{c|}{Search} & \multicolumn{1}{c|}{Update}
            & \multicolumn{1}{c|}{Search}  & \multicolumn{1}{c|}{Update} & \multicolumn{1}{c|}{Search RT}
            & \multicolumn{1}{c|}{Storage} & \multicolumn{1}{c|}{} & \multicolumn{1}{c|}{client} & \multicolumn{1}{c|}{} \\ \hline
      NIMC- SSE\cite{sun2020non}               & $O(a_w)$          & -                     & $O(n_w)$                 & -                 & 1             & $O(|W|\log D))$                  &\xmark                                            & \cmark     &\cmark \\ \hline
       Diana$_{del}$ \cite{bost2017forward}    & $O(a_w) $               & $O(\log a_w)$   & $O(n_w+d_w\log a_w)$     & $O(1)$            & 2             & $O(|W|\log D))$      &Type-\uppercase\expandafter{\romannumeral3}      & \xmark     &\xmark  \\ \hline
      Mitra \cite{ghareh2018new}               & $O(a_w)$                & $O(1)$          & $O(a_w)$                 & $O(1) $           & 2             & $O(|W|\log D))$      &Type-\uppercase\expandafter{\romannumeral2}       & \xmark     &\xmark \\ \hline
      Orion \cite{ghareh2018new}               & $O(n_w\log^2K) $        & $O(\log^2K)$    & $O(n_w\log^2K)  $        & $O(\log^2K)$      & $O(\log K)$   & $O(1)$        &Type-\uppercase\expandafter{\romannumeral1}       & \xmark     &\xmark \\ \hline
      Janus++\cite{10.1145/3243734.3243782}    & $O(n_wd) $              & $O(d)$          & $O(n_w) $                & $O(1)$            & 1             & $O(|W|\log D))$        &Type-\uppercase\expandafter{\romannumeral3}       & \xmark     &\xmark  \\ \hline
      CLOSE-FB\cite{9237959}                   & $O(a_w+CLen) $          & $O(CLen)$       & $O(a_w) $                & $O(1)$            & 2             & $O(1)$        &Type-\uppercase\expandafter{\romannumeral2}       & \xmark     &\xmark  \\ \hline
      NIMS                                     & $O(a_w + |W|) $         & $O(1)$          & $O(a_w) $                & $O(1)$            & 2             & $O(1)$        &Type-\uppercase\expandafter{\romannumeral3}       & {{\color{red}\ding{51}}}     &{{\color{red}\ding{51}}}  \\ \hline
    \end{tabular}
  }
  \label{comparison}
\end{table*}

In most  DSSE schemes \cite{8329523,ghareh2018new,10.1145/3243734.3243782},
the deletion tokens are  generated based on the keyword/document pairs.
As a result,  to delete a single file, the client has to submit  a large number of delete queries if the document contains multiple keywords.
Meanwhile, the data owner has to employ an inverted index or a forward index to store all keyword/document pairs of the whole database.
It is contrary to the data owner's initial goal of outsourcing its storage and calculation services to the cloud server.

To address above issues, we propose a \textbf{N}on-\textbf{I}nteractive \textbf{M}ulti-client DSSE scheme with \textbf{S}mall client storage (NIMS)  via a series of novel designs.
Different from traditional structures (e.g., the structure based on an inverted index),  only a global variable is exploited to reduce client storage.
Existing non-interactive SSE schemes leverage public key encryption with keyword search (PEKS) or constrained pseudorandom function, which results in heavy computation costs.
We notice that clients' time is naturally synchronized. 
Hence, by employing timestamps as an additional input, forward privacy can be achieved without requiring the data owner to be online. 
We use timestamps as an additional input to implement non-interactive forward privacy.

To the best of our knowledge, NIMS is the first non-interactive multi-client DSSE scheme with forward and backward privacy. Besides, our scheme supports small client storage and efficient deletion operations. Our main contributions can be summarized as follows:
\begin{itemize}[leftmargin=*]
\item First, we eliminate the frequently interactive process between the data owner and data users.
We propose a hidden key technique to encrypt the head block key of all keyword chains.
We leverage the time range query to achieve non-interactive forward privacy, since the latest timestamp cannot be searched by the earlier time range.
In addition, we design a novel method that converts the time range query to the wildcard matching problem, and then to a vector dot-product problem, for the highly efficient time range query.

\item Second, our scheme realizes small client storage and can provide high deletion efficiency.
To encode all keyword/document pairs, We  design a Lightweight Storage Chain (LSC) structure,  which only keeps a global variable, so that the storage cost is reduced significantly.
In the LSC structure,  the document identifiers are independently encrypted. Therefore, the data owner can perform deletion operations by sending delete tokens, improving the deletion efficiency.

\item We analyze the security of the proposed scheme and prove that NIMS is IND-CPA security.
We implemented NIMS and conducted extensive experiments to show that NIMS is more computation efficient on a real-world dataset than previous methods.
\end{itemize}

Table~\ref{comparison} shows a comparison between NIMS with previous methods.
It can be observed that our scheme achieves nearly optimal complexity in the search and update.

\section{PROBLEM FORMULATION}\label{section2}
In this section, we describe the system model, the assumed security threats, and the design goals of NIMS.
\subsection{System Model}

\begin{figure}[htbp]
  \centering
  \includegraphics[width=0.47\textwidth]{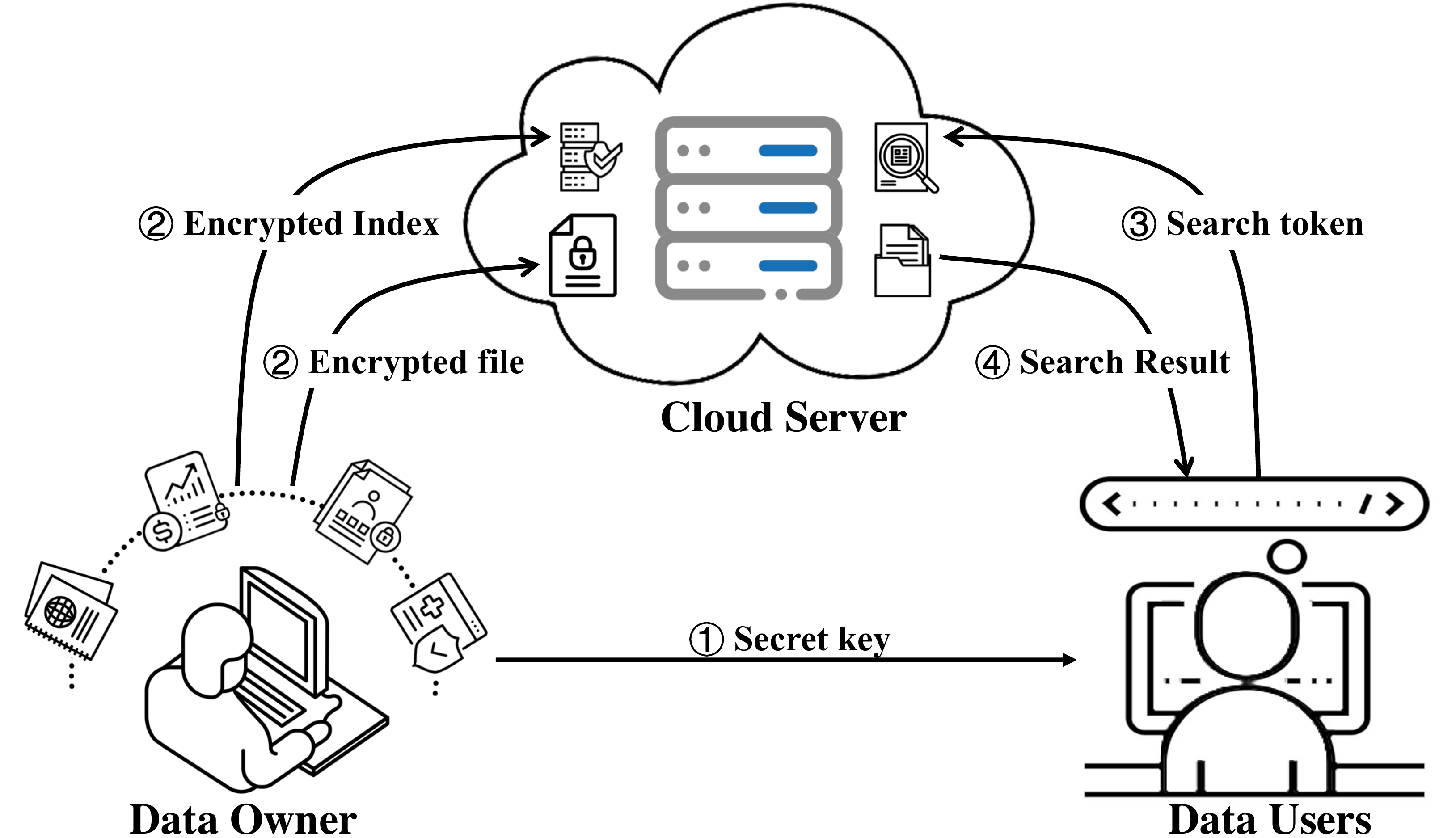}
  \caption{System model.}
  \label{fig:system}
\end{figure}

As shown in Fig.\ref{fig:system}, a multi-client DSSE considers three types of entities: data owner, data users and cloud server.

\begin{itemize}
  \item \textit{Data Owner}: The data owner has a document set $DB$ that contains several keywords. 
  These documents are outsourced to the cloud server after encryption.
  For more efficient utilization and query, the data owner tends to build secure indexes based on keywords.
  In addition, it has the right to update the document at any time.
  After receiving the registration information from a data user, the data owner shares secret keys with it over a secure channel.
  \item \textit{Data Users}: Data users are entitled to access the data by the authorization of the data owner.
  After receiving secret keys, data users have the ability to send query requests to server.
  Data users can obtain the query results set containing the specific keyword and then recover the underlying data via the symmetric keys.
  \item \textit{Server}: With high processing capability and large storage space, the cloud server stores the encrypted documents and indices outsourced by the data owner.
  Besides, the server executes a secure search calculation over encrypted data after receiving a query request and returns the corresponding search result to the data user.
\end{itemize}

\subsection{Threat Model}
\begin{itemize}
  \item \textit{Data Owner}: In the security model, we assume that the data owner can be fully \textit{trusted} since the data is considered personal property. The data owner will follow the protocols sincerely and protect her secret keys securely.
  \item \textit{Data Users}: For data users, we assume that they are \textit{honest}, i.e., they launch the keyword query request faithfully.
  They are not allowed to collude with the server because they pay for data assets and they have to protect data security.
  \item \textit{Server}: The cloud server is considered as \textit{honest-but-curious}.
  The server will faithfully execute  storage and query protocols, but it may be curious about additional information from the encrypted database and query requests.
   Furthermore, the server does not perform active attacks, i.e., colluding with the data owner or data users.
\end{itemize}

\subsection{Design Goals}
Our design goal is to develop a non-interactive multi-client DSSE scheme with small client storage.
Specifically, the following goals should be met.
\begin{itemize}
  \item \textbf{Small client storage:}
  In order to achieve efficient update and privacy requirements, heavy storage costs will be incurred due to storing complex structures.
  Hence, we aim to design a structure with small client storage in our proposed scheme.
  \item \textbf{Efficient deletion:}
  If the deletion operation takes keyword/document pair as input, the computation and communication overhead can be enlarged as the number of keywords in the deleted document increases.
  The solution should minimize the time complexity and communication cost of the delete operation.
  \item \textbf{Non-interaction:} To efficiently update the encrypted database, the data owner maintains a list of keyword states.
  To avoid per-query state transfer between data users and the data owner, we aim to realize non-interactive queries and the search token should be independently generated by a single user.
  \item \textbf{Efficient search:}
 Our scheme attempts to reduce communication and computing costs during keyword search queries to their lowest possible levels. We will improve query efficiency as much as possible.
  \item \textbf{Privacy preservation:}
  The essential requirement of DSSE is privacy preservation, and we need to prevent the server from learning sensitive information.
  Encrypted documents and encrypted indices should be strongly protected.
  Given a search token, the content of query requests should not be learned by the server.
  Even for the same query request, the search token should look random rather than deterministic.
\end{itemize}

\section{Preliminaries}\label{section3}
This section provides some background knowledge and notations utilized in our scheme.

\subsection{Notaions}
In our scheme, we take $\lambda \in \mathbb{N}$ to indicate the security parameter and we denote a negligible function $negl(\lambda)$ in the $\lambda$.
Consider $\{0,1 \}^{l}$ as the collection of binary strings with the length of $l$.
Only the algorithms with polynomial time complexity with respect to $\lambda$ are considered, where we consider an adversary with Probabilistic Polynomial-Time (PPT) algorithms.

The operation $x\stackrel{\$}{\leftarrow} X$ stands for the random sampling from the finite set $X$.
The operator `$||$' and $|X|$ stand for the concatenation of strings and the number of elements of the set $X$, respectively.

Table \ref{tab1} contains a list of additional notations that are required.

\begin{table}[htbp]
  \caption{Notations for describing our DSSE scheme.}
  \begin{center}
    \begin{tabularx}{0.48\textwidth}{lX}
      \hline
      \textbf{Notation}   & \textbf{Description} \\ \hline
      $\enc()$       & the symmetric encryption algorithm  \\
      $\dec()$       & the symmetric decryption algorithm  \\
      $\mathcal{M}_1$,$\mathcal{M}_2$      & secret keys to encrypt matrix  \\
      $DB$          &a database composed of a tuple of documents $DB=\{ doc_i \}_{i=1}^D$\\
      $Docs$         & a document set in one update $Docs \subseteq DB$ \\
      $doc$         & a document $doc=(ind,W_{ind})$ \\
      $ind$         & the identifier of the document $doc$\\
      $W_{ind}$     & a contained set of keywords extracted from $doc$ \\
      $W$           & the keyword set of $DB$, $W=\bigcup_{i=1}^D W_{ind_i}$ \\
      $\iota$       & keyword vector size\\
      $\kappa$      & the length of time vector\\
      $n$           & the size of index and trapdoor vectors,  $n=\iota+\kappa+2$  \\
      $l$           & the number of matrices in the trapdoor set  \\
      $ts$         & the timestamp when index are generated \\
      $tr$         & a time range, from the initial time to the current time \\
      $ctr$         & a global variable, the update counter\\
      $key_w$       & the head block key of the keyword chain $\mathbb{C}_w$\\
      \hline
    \end{tabularx}
    \label{tab1}
  \end{center}
\end{table}

\subsection{Non-interactive Multi-Client DSSE}

A multi-client DSSE scheme $\Pi$ is a triple $\Pi$=(Setup, Search, Update) containing three polynomial-time protocols and involving three entities: data owner, server, and data users. 
\begin{itemize}
  \item \textbf{Setup}$(1^{\lambda}) \rightarrow (msk,\sigma$; EDB):
  This protocol is used to initialize the system.
   The data owner takes as input a security parameter $\lambda$ and outputs ($msk,\sigma$; EDB), where $msk$ is a master secret key, $\sigma$ is the keyword state, and EDB is the initialized encrypted database.
  Finally, the data owner sends $msk$ to data users through a secure channel and sends EDB to the server.
  \item \textbf{Search} ($msk$, $w$; EDB) $\rightarrow$ (DB($w$); EDB$'$):
  This is a protocol between data users and server.
  The data user generates a search token with input ($msk$, $w$) and sends it to the server.
  After execution of the protocol, the server outputs the results as the document set DB($w$) containing the keyword $w$.
  \item \textbf{Update}($op$, $Docs$, $\sigma$; EDB) $\rightarrow$ ($\sigma'$; EDB$'$) : This protocol is run between data owner and server to add (or delete) an entry.
  The data owner takes a document set $Docs$ (or the identifier of the deleted document $ind$) an update operation $op=\{add,del\}$, and a state $\sigma$ as input.
  The input to the server is EDB.
  After running the protocol, the data owner outputs the new state $\sigma'$ and the server modifies EDB.

\end{itemize}

\textbf{Correctness}: If the search protocol consistently provides the correct result DB($w$) for each query, we consider the DSSE scheme is correct.

We also say a DSSE scheme is \textit{non-interactive} if in the search protocol the data user has not interacted with the data user.

\textbf{Security}:
The security of a DSSE scheme should ensure that the server should learn as little information as possible about the encrypted database and queries.
The leaked information is parameterized by the leakage function $\mathcal{L}=(\mathcal{L}_{Setup},\mathcal{L}_{Search},\mathcal{L}_{Update})$
which expresses what information is revealed to the adversary in each protocol.
The definition ensures that a secure DSSE scheme should reveal nothing except what is inferred from the leakage function.

An adversary $\mathcal{A}$ can use its own choosing parameters to trigger protocols at will.
Then, she observes the scenario execution process from the server's perspective and it can obtain transcripts of each operation.
The target of the adversary is to distinguish between a real world \textit{Real} and an ideal world \textit{Ideal} \cite{kamara2012dynamic,bost2017forward}.

\begin{definition}\textit{(Adaptive Secrutiy of DSSE)}\label{Adaptive Secrutiy of DSSE}
A DSSE scheme $\Pi $= (Setup, Search, Update) is \textit{adaptive-secure} with respect to the leakage funciton $\mathcal{L}$, if for any PPT adversary $\mathcal{A}$ issuing polynomial number of queries $q(\lambda)$, there exists a stateful PPT simulator $\mathcal{S}$ such that the following equation holds:
$$|Pr[Real^{\Pi}_{\mathcal{A}}(\lambda)=1]-Pr[Ideal^{\Pi}_{\mathcal{A},\mathcal{S},\mathcal{L}}(\lambda)]| \leq negl(\lambda)$$
\end{definition}
where $Real^{\Pi}_{\mathcal{A}}(\lambda)$ and $Ideal^{\Pi}_{\mathcal{A},\mathcal{S},\mathcal{L}}(\lambda)$ are defined as
\begin{itemize}
    \item $Real^{\Pi}_{\mathcal{A}}(\lambda)$: In the \textit{Real} world, the DSSE scheme is executed in the real case. The adversary $\mathcal{A}$ first selects a database DB and gets an initial encrypted database EDB by running \textit{Setup($1^{\lambda}$)}. 
    Then, $\mathcal{A}$ adaptively performs the search and update protocols depending on the queries $q_i$ and observes the real transcript of all operations. Finally, $\mathcal{A}$  outputs a bit $b$.
    \item $Ideal^{\Pi}_{\mathcal{A},\mathcal{S},\mathcal{L}}(\lambda)$: In the \textit{Ideal} world, $\mathcal{A}$ obtains a simulated transcript generated by a PPT simulator with leakage functions $\mathcal{L}$, instead of the real transcript. After choosing a database DB, $\mathcal{A}$ is given an encrypted database generated by  $\mathcal{S}(\mathcal{L}_{Setup})$. The adversary repeatedly performs search and update queries. $\mathcal{A}$ receives the transcripts $\mathcal{S}(\mathcal{L}_{Search})$ and $\mathcal{S}(\mathcal{L}_{Update})$ generated by the simulator.
    Eventually, $\mathcal{A}$ outputs a bit 0 or 1.
\end{itemize}

\subsection{Leakage Function}
The leakage function $\mathcal{L}$ describes the information revealed during executing protocols.
$\mathcal{L}$ keeps all previous queries  as a state in a list $Q$ which records all timestamps corresponding to the operations for a keyword $w$.
Each entry of the query list $Q$ is a pair ($t$, $w$), where $w$ is a keyword, and $t$ is a timestamp increasing with each query.

Repeated execution of the query will be leaked some standard search leakage types include \textit{search pattern} and \textit{access pattern}.
Search pattern $sp(w)=\{t|(t,w)\in Q\}$ reveals the relevance of searches on the same keyword $w$.
Access pattern $ap(w)=\{DB(w)\}$ reveals the result of searching for $w$.
Note that the access pattern leakage is inevitable (unless using an oblivious RAM) if the client expects to obtain the entity files matched, not just their identifiers.

We additionally provide some leakage functions related to backward privacy.
TimeDB($w$) is a list of inserting timestamps for all documents containing the keyword $w$ that have not been deleted subsequently.
TimeDB($w$) = $\{(t,id)|(t,add,(w,id))\in Q$ and $\forall$ $t',(t',del,(w,id)) \notin Q \}$.
DelHist($w$) is a list of (addition timestamp, deletion timestamp) pairs of documents that have been deleted.
DelHist($w$) = $\{(t^{add}, t^{del})| \exists id:(t^{add}, add, (w,id))\in Q \wedge (t^{del}, del, (w,id)) \in Q) \}$.

\subsection{Forward and Backward Privacy}
Forward and backward privacy of DSSE was first formally defined by Bost et al.\cite{bost2017forward}, which is particularly useful in practice.

\subsubsection{Forward Privacy}
A DSSE scheme is forward private if the newly added document cannot be connected with a  previously searched keyword.
We follow the formal definition in \cite{bost2016ovarphiovarsigma}.
\begin{definition}\label{de_forward}(Forward Privacy)
An $\mathcal{L}$-adaptively-secure DSSE scheme is forward private if the update leakage function $\mathcal{L}_{Update}$ can be written as
$$\mathcal{L}_{Update}(op,Docs)=\mathcal{L}'(op,\{|ind_i|,|W_{ind_i}|\}),$$
where $\mathcal{L}'$ is a stateless function, $op=add/del$, and the set $\{|ind_i|,|W_{ind_i}|\}$ is the number of identifiers and keywords of all updated documents in $Docs$.
\end{definition}

\subsubsection{Backward Privacy}

Backward privacy ensures that if a document $doc$ containing the keyword $w$ is deleted, the subsequent search queries on keyword $w$ should reveal nothing about $doc$. 
Backward privacy is defined as three types according to different leakage types: Type-I to Type-\uppercase\expandafter{\romannumeral3}.
Type-I reveals the least information, whereas Type-III reveals the most.
Using the above functions, Type-\uppercase\expandafter{\romannumeral3} backward privacy is defined as follows.
\begin{definition} (Type-\uppercase\expandafter{\romannumeral3} Backward privacy)
A $\mathcal{L}$-adaptively-secure SSE scheme is Type-\uppercase\expandafter{\romannumeral3} backward private:
\textit{iff} $\mathcal{L}_{Update}$($op$, $Docs$) $=$ $\mathcal{L}'$($op$, $\{ |ind_i|, |W_{ind_i}| \}$) and
$\mathcal{L}_{Search}$($w$) $=$ $\mathcal{L}''$(TimeDB($w$), DelHist($w$)),
where $\mathcal{L}'$ and $\mathcal{L}''$ are stateless functions.
\end{definition}

\section{Lightweight Storage Chain Structure}\label{ICS}
In this section, we present the Lightweight Storage Chain (LSC) structure to build encrypted chains, thus we can store all keyword/document pairs at the server securely.
Most existing DSSE schemes with forward privacy \cite{yang2017rspp,li2019searchable,8329523} have to record all keyword statuses, such as the keyword counter \cite{8329523,ghareh2018new}.
When a data user performs a search query, it needs to interact with the data owner to ask for the latest keyword status to help it generate the search token.
Meanwhile, to execute the delete operation, many backward private schemes \cite{zheng2020achieving,10.1007/978-3-030-88428-4_1,8329523,9724186} takes ($ind$, $w$) pairs as input, thus the data owner has to maintain a complex structure (inverted index or forward index) to record all keyword/document pairs locally.
As a result, the storage consumption of the data owner increases as the number of keyword/document pairs increases.
To address this issue, we develop the LSC structure, which implicitly links all keyword/document pairs and takes small client-side storage.
The data owner only needs to maintain an update counter and the keyword set locally.
The keyword set and an update counter are the only objects the data owner locally kept up.
With this structure, the permanent client storage is $O(1)$.

\begin{figure}[htbp]
  \centering
  \includegraphics[width=0.49\textwidth]{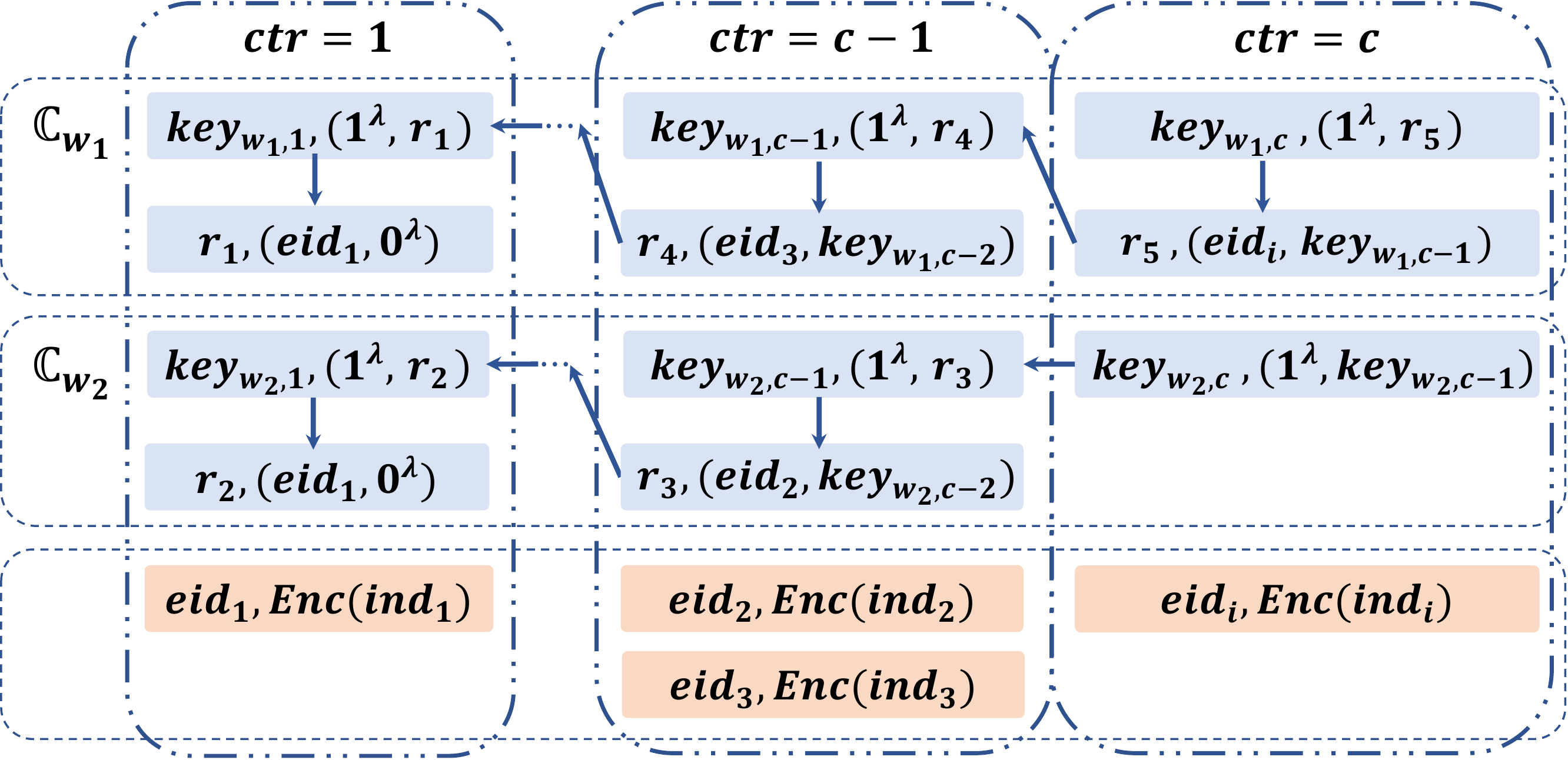}
  \caption{Example of Lightweight Storage Chain Structure.}
  \label{fig:chain}
\end{figure}
For better understanding, we first explain the chain structure under plaintext.
We define a current block $(key, (data, kpr))$, where $key$ is the encryption key and $(data, kpr)$ is the plaintext to be encrypted by $key$.
By setting $kpr$ as the encryption key of the block previous to the current block, the current block is implicitly connected to the previous block.
Let $\mathbb{C}_w$ be a keyword chain containing data blocks corresponding to the same keyword $w$ together.
The last block to be added to the keyword chain $\mathbb{C}_w$ in a single update is referred to as the head block, and its encryption key $key$ is referred to as the head block key.
The head block key can be calculated $\pseudo(ctr,w)$ by using a pseudorandom function.

As shown in Fig.\ref{fig:chain}, we give an example to show how the LSC structure works.
In our scheme, each keyword $w$ corresponds to a chain $\mathbb{C}_w$ and we set a global counter $ctr$ as the update counter.
We assume the keyword set $W$ contains two keywords $\{ w_1,w_2\}$ corresponding to two keyword chains $\mathbb{C}_{w_1}$ and $\mathbb{C}_{w_2}$, respectively.
We describe in detail the procedures for the $c$-th update which adds a document set  $Docs=\{doc_i\}$, where $doc_i = \{ind_i,\{w_1\}\}$.
\begin{itemize}
  \item First, the data owner calculates the head block key $key_{w_1,c-1}$ of the  chain $\mathbb{C}_{w_1}$ and $key_{w_2,c-1}$ of the chain $\mathbb{C}_{w_2}$.
  \item Second, the data owner computes $\enc(ind_i)$ by encrypting the identifier $ind_i$ and generates an address $eid_i$ for it.
  Then, $eid_i$ and $\enc(ind_i)$ form a two-tuple.
  \item Third, for the keyword $w_1$, the data owner randomly chooses a block key $r_5$, and constructs a block $(r_5,(eid_i,key_{w_1,c-1}))$ which implicitly linked to the previous block.
  \item Fourth, after the document $doc_i$ is encoded, the data owner calculates a head block key $key_{w_1,c}$ for the keyword chain $\mathbb{C}_{w_1}$ and $key_{w_2,c}$ for the keyword chain $\mathbb{C}_{w_2}$. Then, she constructs two blocks $(key_{w_1,c},(1^{\lambda},r_5))$ and $(key_{w_2,c},(1^{\lambda},key_{w_2,c-1}))$.
  \item Finally, all blocks are sent to the server for storage, along with the $(eid_i,\enc(ind_i))$ pairs.
\end{itemize}

By using the structure, the data owner only stores a global counter $ctr$.
To delete a document $doc_i$, the data owner calculate the address $eid_i$ of the encrypted identifier as the delete token and the server removes the corresponding $(eid_i,Enc(ind_i))$ pair.

Next, we describe how to encrypt the block and retrieve the keyword chain.
\begin{itemize}
  \item \textbf{EncryptBlock}($\mathbb{C}, key, data, kpr$): The data owner performs this algorithm to encrypt the block with taking ($key,data,kpr$) as input.
  The data owner computes $val$ = $\hash(key||1)\bigoplus (data||kpr)$ and sets an address $addr$ = $\hash(key||0)$. Finally, the block $(addr, val)$ are added to $\mathbb{C}$.
  \item \textbf{RetrieveChain}($\mathbb{C}, key$):
  The server takes the head block key $key$ of a chain $\mathbb{C}$ as input and retrieves all data blocks of $\mathbb{C}$.
  It consists of three steps: 1) finding the block $b = (addr,val)$ by calculating $addr = \hash(key||0)$, 2) decrypting $val$ and recover the previous block key by computing $\hash(key||1)\bigoplus val$, and 3) setting $key$ to $kpr$ and repeat the step 1 and 2 until $kpr=0^\lambda$.
\end{itemize}

\begin{algorithm*}[htbp]{\vspace{-0.5cm}}
	\caption{Boolean Wildcard Matching Algorithm (BWMA)}
	\label{alg:algorithm1}
	\begin{multicols}{2}
	\underline{\textbf{TransIndex}($\overrightarrow{P},m$)}
	\begin{algorithmic}[1]
		\State Given a specific boolean vector $\overrightarrow{P}=\{0,1\}^m$ , initialize a vector $\overrightarrow{P}_t = [\tilde{p}_0, \cdots, \tilde{p}_m, \tilde{p}_{m+1}]$.
        \For{$i = 0 $ to $m$}
            \If{$p_i = 0$}
                \State $\tilde{p}_i \leftarrow 1$
            \ElsIf{$q_i = 1$}
                \State $\tilde{p}_i \leftarrow -1$
            \EndIf
        \EndFor
        \State $\tilde{p}_{m+1} \leftarrow 1$
        \State \Return $\overrightarrow{P}_t$
	\end{algorithmic}

	\underline{\textbf{TransQuery}($\overrightarrow{Q}, m$)}
	\begin{algorithmic}[1]
		\State Given a specific vector $\overrightarrow{Q}=\{0,1,*\}^m$, initialize a vector $\overrightarrow{Q}_t = [\tilde{q}_0, \cdots, \tilde{q}_m, \tilde{q}_{m+1}]$ and a counter $cnt=0$.
        \For{$i = 0$ to $m$}
            \If{$q_i = 0$}
                \State $\tilde{q}_i \leftarrow 1$
            \ElsIf{$q_i = 1$}
                \State $\tilde{q}_i \leftarrow -1$
            \ElsIf{$q_i = *$}
                \State $\tilde{q}_i \leftarrow 0$
                \State $cnt \leftarrow cnt + 1$
            \EndIf
        \EndFor
        \State $\tilde{q}_{m+1} \leftarrow (m-cnt)$
        \State \Return $\overrightarrow{Q}_t$
	\end{algorithmic}

	\underline{\textbf{Match}($\overrightarrow{P}_t$, $\overrightarrow{Q}_t)$}
	\begin{algorithmic}[1]
	
        \State res $\leftarrow$ $\overrightarrow{P}_t \cdot \overrightarrow{Q}_t$
        \If{res = 0}
	        \State \Return True. \Comment $\overrightarrow{P}$ and $\overrightarrow{Q}$ are matched.
	   \Else
	        \State \Return False. \Comment $\overrightarrow{P}$ and $\overrightarrow{Q}$ are not matched.
	   \EndIf
	\end{algorithmic}
	\end{multicols}
 {\vspace{-0.4cm}}
\end{algorithm*}

The client can update the chain $\mathbb{C}$ by adding a new block and search the index by using the head block key.
Based on the LSC structure, the data owner only needs to store a global counter instead of all keyword/document pairs to reduce the storage cost.
Besides, it needs to submit a delete token $eid$ to delete a document which improves the deletion efficiency.
Moreover, all blocks are only logically linked together.
Meanwhile, the cloud server cannot tell which chain a newly added block belongs to, thus privacy preservation is achieved.

\section{Hidden Key Technique}\label{HKT_sec}
\subsection{Overview}
Based on the LSC structure, the data owner stores a global counter locally.
In order to retrieve documents containing a keyword $w$, the data user needs to get the global variable or the head block key of the searched keyword chain $\mathbb{C}_{w}$.
A simple method is that the data owner actively delivers the counter (or all head block keys) after each update or the user asks for the search token.
Thus, frequently interactive processes are required, which causes heavy communication costs and leaks more information.
Another method is that the update counter $ctr$ is encrypted and stored  on the server.
Thus, the data user needs to request the encrypted counter, decrypt it, produce a search token, and return it to the server.
As a result, two rounds of interaction are required.
In addition, forward privacy cannot be achieved.

To address this issue, we propose the indices (HK) technique to support non-interactive queries while protecting more information and providing forward privacy.
In the HK technique, the data owner hides all head block keys $\{key_{w_i,ctr}\}$ into the ciphertext, and the data user queries the ciphertext to get the required head block key.
Meanwhile, it should be ensured that the search token, which is independently generated by the data user, can only search previously added keys rather than later-added keys.
To achieve non-interactive forward privacy, we introduce the time range query since the time is multi-client synchronization without interaction and the latest timestamp cannot be searched by earlier time ranges.
Specifically, the data owner adds a timestamp on the head block key.
In a query, a data user generates a time range from the initial time to the current time and submits it with the searched keyword to query the key.
To perform the time range query, we first convert the range query problem to a boolean wildcard matching problem and then design the Boolean Wildcard Matching Algorithm (BWMA) to settle the problem.

For ease of understanding, we first describe how the BWMA works, and then we introduce how the time range query under plaintext.
Finally, we employ the time range query to construct the HK technique, a multi-client non-interactive method for querying the head block key.

\subsection{BWMA}
The main idea of BWMA is that converts the boolean wildcard vector matching problem into a dot product calculation problem.
The BWMA can be adopted widely in searchable encryption scenarios, such as subset-query and range-query.
BWMA is defined using three PPT algorithms.
The details are shown in Algorithm~\ref{alg:algorithm1}.

For instance, given a boolean vector $\overrightarrow{P}=[1,0,1,1]$ and a boolean wildcard vector $\overrightarrow{Q}=[1,0,*,*]$, we transform $\overrightarrow{P}$ to $ \overrightarrow{P}_t=[-1,1,-1,-1,1]$ using TransIndex($\cdot$) and transform
$\overrightarrow{Q}=[1,0,*,*]$ to $ \overrightarrow{Q}_t=[-1,1,0,0,-2]$ using TransQuery($\cdot$).
Finally, we check if there is a match between two vectors by evaluating the inner product.

\begin{algorithm*}[htbp]{\vspace{-0.4cm}}
	\caption{Hidden key technique}\label{alg:HKT}
	\begin{multicols}{2}
	\underline{$\mathcal{U}_w^*$ $\leftarrow$ \textbf{HKData}($w$, $Ts$, $key_w$, $\mathcal{M}_1$, $\mathcal{M}_2$)}
	    \begin{algorithmic}[1] 
            \State $\overrightarrow{H}$ $\leftarrow$ Binary($\hash(w)$, $\iota$)
            \State $\overrightarrow{T}$ $\leftarrow$ Binary($Ts$, $\kappa$)
            \State $\overrightarrow{P}_t$ $\leftarrow$
            BWMA.TransIndex($\overrightarrow{H}$ + $\overrightarrow{T}$, $n-2$)
            \State $\overrightarrow{U}[1:n-1]$ $\leftarrow$ $r_u\cdot \overrightarrow{P}_t$, $\overrightarrow{U}[n]$ $\leftarrow$ $key_w$, where $r_u$ is a random number larger than $key_w$.
            \State $\mathcal{U}$ $\leftarrow$  GenLowTriMart($\overrightarrow{U}$)
            \State $\mathcal{I}_x$ $\leftarrow$  GenLowTriMart($\overrightarrow{I}$), where $\overrightarrow{I}$ is a $n$-dimensional vector with all elements being 1
            \State $\mathcal{U}_w^*$ $\leftarrow$ $\mathcal{M}_1 \times \mathcal{I}_x \times \mathcal{U} \times \mathcal{I}_x \times \mathcal{M}_2$, where $\mathcal{M}_1$ and $\mathcal{M}_2$ are invertible matrices.
            \State \Return  $\mathcal{U}_w^*$
        \end{algorithmic}
	\underline{$\{\mathcal{Q}_m^*\}_{m=1}^l$ $\leftarrow$ \textbf{HKToken}($w$, $Tr$, $\mathcal{M}_1$, $\mathcal{M}_2$)}
	    \begin{algorithmic}[1] 
            \State $\overrightarrow{H}$ $\leftarrow$ Binary($\hash(w)$, $\iota$)
            \State $\{\overrightarrow{T_m} \}^l_{m=1}$ $\leftarrow$ Wildcard($Tr$, $\kappa$)
            \For{m = $1$ to $l$}
                \State $\overrightarrow{G}_t$ $\leftarrow$
                 BWMA.TransQuery($\overrightarrow{H}$ + $\overrightarrow{T_m}$, $n-2$)
                \State $\overrightarrow{Q_m}[1 : n-1]$ $\leftarrow$ $r_m\cdot \overrightarrow{G}_t$,  $\overrightarrow{Q_m}[n] \leftarrow 1$, where $r_m$ is a random number
                \State $\mathcal{Q}_m$ $\leftarrow$ GenLowTriMart($\overrightarrow{Q_m}$)
                \State $\mathcal{I}_y$ $\leftarrow$ GenLowTriMart($\overrightarrow{I}$)
                \State $\mathcal{Q}_m^* \leftarrow \mathcal{M}_2^{-1}\times \mathcal{I}_y \times \mathcal{Q}_m \times \mathcal{I}_y \times \mathcal{M}_1^{-1}$, where $\mathcal{M}_1^{-1}$ and $\mathcal{M}_2^{-1}$ are the inverse matrices of the secret keys.
            \EndFor
            \State \Return  $\{\mathcal{Q}_m^*\}_{m=1}^l$
        \end{algorithmic}
	\underline{$flag$, $key_w$ $\leftarrow$ \textbf{HKQuery}($\mathcal{U}_w^*$,  $\{\mathcal{Q}_m^*\}_{m=1}^l)$}
	    \begin{algorithmic}[1] 
            \State $flag \leftarrow 0$
            \For{m = $1$ to $l$}
                \State $res \leftarrow tr(\mathcal{U}_w^* \times \mathcal{Q}_m^*)$
                \If{$res$ $\textgreater$ 0}
                    \State $key_w \leftarrow res$, $flag \leftarrow 1$
                    \State break
                \EndIf
            \EndFor
            \State \Return $flag$, $st_w$
        \end{algorithmic}
	\end{multicols}
 {\vspace{-0.3cm}}
\end{algorithm*}

\subsection{Time range query}
Assuming there are a timestamp $T$ and a time range $[T_{min},T_{max}]$, we will check whether the timestamp is within the time range.
In order to improve the matching efficiency, we utilize a binary tree to decrease the amount of elements in range queries.
Specifically, we construct a perfect binary tree with $2^\kappa$ leaf nodes and each node is regarded as a $\kappa$-bit value.
The timestamp $T$ is regarded as a leaf node and its value is denoted as a $\kappa$-bits vector.
$T_{min}$ and $T_{max}$ are considered as two leaf nodes and the range is replaced by the minimum set that can cover the range.
For example, as shown in Fig.\ref{fig:tr}, we check whether timestamp $T=3$ is within the time range $T_r=[0,5]$.
The timestamp $T$ (the green node) is presented as $\overrightarrow{T}=[0,1,1]$ and the time range (the orange nodes) is denoted as $S(T_r)=\{[0,*,*],[1,0,*]\}$.

\begin{figure}[htbp]
  \centering
  \includegraphics[width=0.42\textwidth]{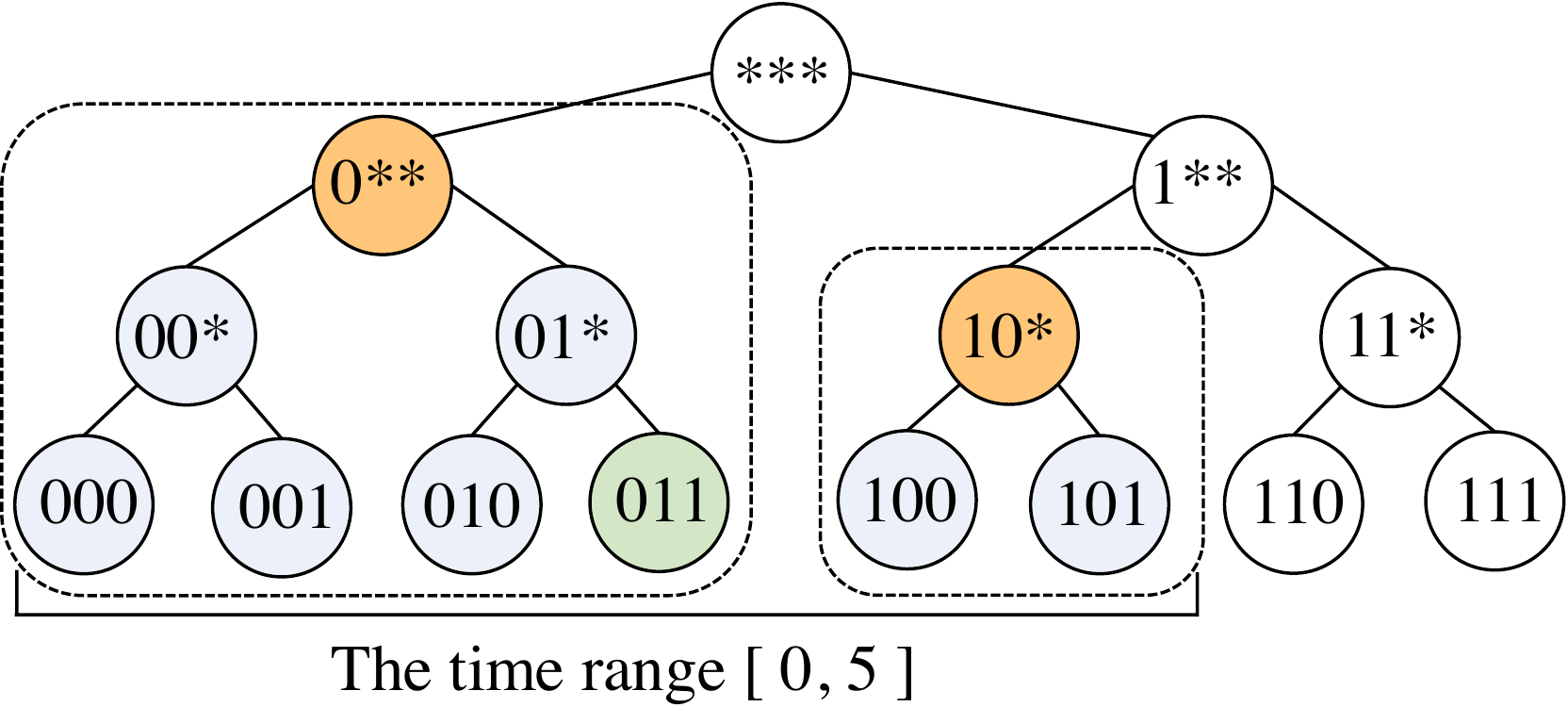}
  \caption{Example of time range query.}
  \label{fig:tr}
\end{figure}

The above transformation enables the range query problem to be formulated as a boolean wildcard matching problem.
In other words, the query is expressed as whether there is a match between the boolean vector $\overrightarrow{T}$ and the wildcard vector set $S(\overrightarrow{T}_r)$.
Finally, we conclude the result by applying BWMA to the timestamp vector and the time range vector.
For example, $\overrightarrow{T}=[0,1,1]$ is transformed to $\overrightarrow{T}_t=[1,-1,-1,1]$ and
$S(T_r)=\{[0,*,*],[1,0,*]\}$ is transformed to $S_t(T_r)=\{\overrightarrow{T_1},\overrightarrow{T_2}\}=\{[1,0,0,-1],[-1,1,0,-2]\}$.
There is a match between the time point $T$ and the time range $T_r$, since the inner product of $\overrightarrow{T}$ and $\overrightarrow{T_1}$ is zero.

\subsection{Details of HK technique}
From a high-level perspective, our design can be interpreted as a two-part search, with the first part being a keyword search and the second part being a time-range query.
If the two parts are matched, the head block key for a certain keyword can be obtained.
In this technique, to encrypt the index vector and token vector,  we leverage the random matrix multiplication technique which enables calculating the inner product of two vectors under the ciphertext.

\begin{figure}[htbp]
  \centering
  \includegraphics[width=0.49\textwidth]{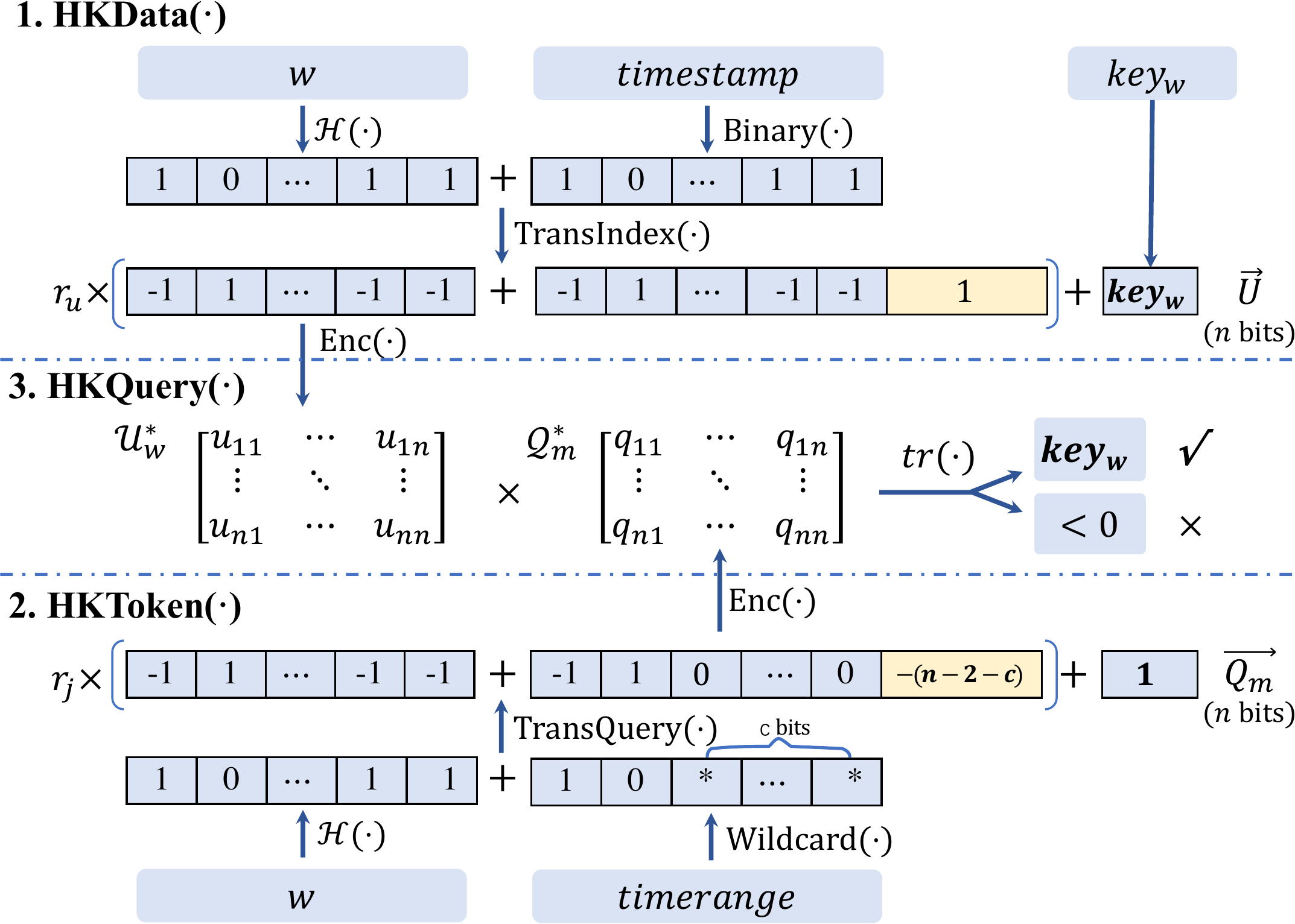}
  \caption{Construction of HKT.}
  \label{fig:HKT}
\end{figure}

We introduce three components of the HK framework with an example as shown in Fig.\ref{fig:HKT}.
The procedures of HK technique are presented in detail in Algorithm \ref{alg:HKT}, where Binary$(D, k)$ is a function that encodes $D$ to a $k$-dimension vector and Wildcard$(R, k)$ is a function that generates a set of wildcard boolean vectors ($k$-dimension) according to a range $R$.
Let GenLowTriMart($\overrightarrow{P}$) be a function that generates a random $n \times n$ lower triangular matrix with the main diagonal being $\overrightarrow{P}=(p_1,p_2,\cdots,p_n)$ as
                $$\mathcal{P} =
                \scriptsize{
            \left[
              \begin{array}{cccc}
                p_1      & 0      & \cdots & 0      \\
                *      & p_2      & \cdots & 0      \\
                \vdots & \vdots & \ddots & \vdots \\
                *      & *      & \cdots & p_n      \\
              \end{array}
              \right].
          }
        $$
where we use $*$ to denote a random value for ease of presentation.
Let $tr(\mathcal{A})$ be the trace of the matrix $\mathcal{A}$.

\textbf{HKData}($\cdot$). In this phase, the data owner hides the head block key $key_w$ in a vector form with the keyword and the timestamp and then encrypts the vector using the random matrix multiplication technique.
Firstly, the data owner transforms the hash value of keyword $w$ and the timestamp $Ts$ to the binary format by executing Binary($\cdot$) encoding, respectively (Line 1-2).
Next, for the concatenation vector of $\overrightarrow{H}$ and $\overrightarrow{T}$, it is encoded by BMWA.TransIndex($\cdot$) (Line 3).
The data owner uses a random number $r_u$ to perturb the vector.
Finally, the key $key_w$ is appended to the last bit of the transformed vector (Line 4) and the vector $\overrightarrow{U}$ is encrypted by using the random matrix multiplication technique  (Line 5-7).
The encryption matrix is sent to the server for storage.

\textbf{HKToken}($\cdot$). In this phase, the data user generates a search token with the queried keyword and a time range from the initial time to the current time.
The process is similar to HKData($\cdot$).
The difference is that the time range is transformed into multiple boolean wildcard vectors and the search token consists of a set of encrypted matrices.

\textbf{HKQuery}($\cdot$).In this phase, the server calculates the trace between the randomly chosen encrypted key and all tokens in the token set.
Only when the keyword is matched and the timestamp is within the time range, the trace will be greater than zero and the result is the head block key of the searched keyword chain.

\textbf{Remark}.
Synchronization of time guarantees that interactions are not needed and newly added timestamps cannot be searched by past time ranges.
Therefore, it is impossible that the previous token searches for the latest head block key, thus achieving forward privacy without interaction.

\section{Scheme Details}\label{section_detail}
In this section, we describe our multi-client non-interactive DSSE scheme with small client storage cost, forward privacy, Type-\uppercase\expandafter{\romannumeral3} backward privacy, and optimal computational efficiency.
\subsection{Overview}
NIMS is comprised of three algorithms: Setup, Update and Search.
In Setup, the data owner initializes the scheme with the system parameter.
In Update, given a document set $Docs$, the data owner encodes all keyword/document pairs using LSC structure based on the update counter $ctr$.
The head block keys of all keyword chains are encrypted to a matrix set $Mat$ via HKData($\cdot$) and stored on the server.
To delete a document with identifier $ind$, the data owner calculates the index address $eid$ by inputing $ind$ and the server removes the $(eid,Enc(ind))$ pair.
In Search, the data user generates a search token set using HKToken($\cdot$) based on the current time.
The server randomly chooses a matrix from $Mat$ and calculates the trace between the matrix and each token matrix in the token set.
Once a trace is greater than zero, the trace is the head block key of the searched keyword.
Then, the server utilizes the head block key to retrieve all of the required nodes in the LSC structure and search for the encrypted document identifier.

\subsection{Construction of NIMS}

\begin{algorithm}
        \caption{$(msk,\sigma ,EDB)$ $\leftarrow$ Setup$(\lambda)$}
        \label{ZSetup}
        \begin{algorithmic}[1] 
            \Require
                \State  Generate $n\times n$ invertible matrices $\mathcal{M}_1$ and $\mathcal{M}_2$ randomly
                \State $\key \stackrel{\$}{\leftarrow} \{0,1 \}^{\lambda}$, $ctr \leftarrow 0$,
                $msk \leftarrow \{\key,\mathcal{M}_1,\mathcal{M}_2\}$
                \State $\sigma \leftarrow (msk, ctr)$
                \State $Mat \leftarrow \varnothing $, $CDB \leftarrow \varnothing$, $EDB \leftarrow (CDB, Mat)$
                \State Send $msk$ to the data users.
                \State Send $EDB$ to the server.
        \end{algorithmic}
\end{algorithm}

Let $\pseudo_1:\{0,1\}^* \times \{0,1\}^*\rightarrow \{0,1\}^\lambda$ and $\pseudo_2: \{0,1\}^\lambda  \times \{0,1\}^*  \rightarrow \{0,1\}^{\lambda}$ be pseudo-random functions.
Consider two hash functions,  $\hash: \{0,1\}^* \rightarrow \{0,1\}^{\iota}$ and $\hash_1: \{0,1\}^* \rightarrow \{0,1\}^{2\lambda}$.

\textbf{Setup:}
Algorithm \ref{ZSetup} gives a formal description.

\begin{itemize}
    \item The data owner first creates a secret key $\key$ for the pseudo-random function $\pseudo_2$ and two $n\times n$ dimensions invertible matrices $\{\mathcal{M}_1,\mathcal{M}_2\}$. A global counter $ctr$ is initialized to 0. The secret keys $msk$ are sent to data users over secure channels.
    \item The data owner initializes an empty table $Mat$ to store the encrypted status matrices and an encrypted index database $CDB$ as a dual dictionary data structure. They are sent to the server to store encrypted entries.
\end{itemize}

\begin{algorithm}
        \caption{$(msk,\sigma;EDB)$ $\leftarrow$ Update($msk$, $add$, $Docs$, $\sigma$; $EDB'$)}
        \label{Zadd}
        \begin{algorithmic}[1] 
            \Require
                \State Parse $\sigma$ as $ctr$, $ctr \leftarrow  ctr +1$, $Nmat \leftarrow \varnothing $, $dic \leftarrow \varnothing$, $Key \leftarrow \varnothing$
                \For{all $w$ in $W$}
                    \State $Key[w] \leftarrow  \pseudo_1(ctr-1, w)$
                \EndFor
                \While{ $|Docs|$ $\neq$ 0}
                    \State $doc \stackrel{\$}{\leftarrow} Docs$
                    \State Parse $doc$ as $(ind, W_{ind})$
                    \State $eid \leftarrow  \pseudo_2(\key, ind)$, $dic \leftarrow dic\cup (eid, \enc(ind))$
                    \While{ $|W_{ind}| \neq 0$}
                        \State $w \stackrel{\$}{\leftarrow} W_{ind}$, $W_{ind} \leftarrow W_{ind} \backslash \{w\}$
		       \If{$w \notin W$}
                     \State $Key[w] \leftarrow 0^\lambda$
			         \State $W \leftarrow W \cup w$
		       \EndIf
                        \State $key \stackrel{\$}{\leftarrow} \{0,1\}^{\lambda}$, $kpr \leftarrow Key[w]$
                        \State $mask \leftarrow \hash_1(key||1)$, $value \leftarrow (eid||kpr)$
                        \State $dic \leftarrow dic\cup (\hash_1(key||0), mask \bigoplus value)$
                        \State $Key[w] \leftarrow key$
                        \State $Docs \leftarrow Docs \backslash \{doc\}$
                    \EndWhile
                \EndWhile
                \For{all $w$ in $W$}
                    \State $key_w \leftarrow \pseudo_1(ctr,w)$, $kpr \leftarrow Key[w]$
                    \State $mask \leftarrow \hash_1(key_w||1)$, $value \leftarrow (1^{\lambda}||kpr)$
                    \State $dic \leftarrow dic\cup (\hash_1(key_w||0), mask \bigoplus value)$
                    \State $mat \leftarrow$ HKData$(w, Ts, key_w, \mathcal{M}_1 \mathcal{M}_2)$
                    \State $Nmat \leftarrow Nmat \cup mat$
                    \State Send $dic$, $Nmat$ to the server
                \EndFor
            \Ensure
                \State $CDB\leftarrow CDB \cup dic$, $Mat \leftarrow Nmat$
                \State $EDB' \leftarrow (CDB, Mat)$
        \end{algorithmic}
\end{algorithm}

\textbf{Add:} In this procedure (Algorithm \ref{Zadd}), the data owner encrypts all keyword/document pairs of the document set $Docs$.

\begin{itemize}
    \item The data owner only holds the state consisting of a secret key $msk$ and an update counter $ctr$ increasing with the update time.
    At initialization, the head block key is uniformly set to be $0^{\lambda}$.
    Otherwise, it first calculates the head block key $Key[w] = \pseudo_1(ctr-1, w)$.
    \item  The data owner randomly picks a document in $Docs$ to generate indices and encrypts the identifier using the symmetric encryption, then it sets an address $eid = \pseudo_2(\key, ind)$ for the identifier $ind$ with the secret key $\key$.
    Each keyword of that document is encoded using LSC structure.
    For a keyword $w\in W_{ind}$, it randomly generates $key$ as the block key, runs the hash function $\hash_1$ twice, and computes $\hash_1(key||0)$ and $\hash_1(key||1)$.
    $\hash_1(key||0)$ is used as the key address and $\hash_1(key||1)$ is XORed with the entry ($eid||kpr$), where $kpr$ is the former block key.
    The utilization of the hash function and XOR implicitly relates all blocks.
    \item After all keyword/document pairs of $Docs$ are updated, for each keyword $w$, the data owner calculates the head block key $key_w = \pseudo_1(ctr, w)$ and then runs the hash function twice to encrypt  $1^{\lambda}$.
    The head block key $key_w$ is encrypted to a matrix $mat$ by HKT.
    The matrix set $Nmat$ is sent to the server with the dictionary $dic$.
    \item The server adds $dic$ to the encrypted database CDB and replaces the matrix set $Mat$ with the received set $Nmat$.
    \item On the server side, the complexity in terms of storage, computation, and communication for the Add operation are all $O(K)$, where $K$ is the number of keyword/document pairs.
\end{itemize}

\begin{algorithm}
        \caption{$(msk,\sigma ,EDB')$ $\leftarrow$ Update($msk$, $delete$, $ind$; $EDB'$)}
        \label{ZDelete}
        \begin{algorithmic}[1] 
            \Require
                \State $eid \leftarrow \pseudo_2(\key, ind)$
                \State Send $eid$ to the server
            \Ensure
                \State $CDB[eid] \leftarrow \varnothing $
        \end{algorithmic}
\end{algorithm}

\textbf{Delete:} The details are described in Algorithm \ref{ZDelete}.
\begin{itemize}
    \item
    To delete a document $doc_{ind}$, a deletion token $eid = \pseudo_2(\key,ind)$ is calculated and sent to the server.
    After receiving the deletion token, the server deletes the entry with address $eid$ (marking it as inaccessible).
    \item Note that the deletion operation only impacts one block in $dic$, because the identifier is not stored in the keyword chain directly.
    Meanwhile, since the data owner only generates one token instead of the number of keywords in the document $doc_{ind}$ ($W_{ind}$), the computation and communication complexity is $O(1)$
\end{itemize}

\begin{algorithm}
        \caption{$DB(w)$ $\leftarrow$ Search$(w,EDB)$}
        \label{ZSearch}
        \begin{algorithmic}[1] 
            \renewcommand{\algorithmicrequire}{\textbf{Data User:}}
            \Require
                \State  $\{\mathcal{Q}_m^*\}_{m=1}^l \leftarrow$ HKToken$(w, Tr, \mathcal{M}_1 \mathcal{M}_2)$
                \State Send $\{\mathcal{Q}_m^*\}_{m=1}^l$ to the Server

            \Ensure
                $\mathcal{X}$ $\leftarrow \varnothing$
                \For{$mat$ in $Mat$}
                    \State $flag, res \leftarrow$ HKQuery($mat$, $\{\mathcal{Q}_m^*\}_{m=1}^l$)
                    \If{$flag$ $=$ 1}
                        \State $key$ $\leftarrow$ $res$
                        \State break
                    \EndIf
                \EndFor

                \While{$key \neq 0^{\lambda}$}
                    \State $addr \leftarrow \hash_1(key||0)$, $mask \leftarrow \hash_1(key||1)$
                    \State $eid||kpr \leftarrow CDB[addr]\bigoplus mask$, $key \leftarrow kpr$
                    \If{$CDB[eid] \neq \varnothing$ and $eid$ $!=$ $1^{\lambda}$}
                        \State $\mathcal{X} \leftarrow \mathcal{X}$ $\cup$ $CDB[eid]$
                    \EndIf
                \EndWhile
                \State Send $\mathcal{X}$ to Data User
            \renewcommand{\algorithmicrequire}{\textbf{Data User:}}
            \Require
                \For{$x$ in $\mathcal{X}$}
                    \State DB($w$) $\leftarrow$ DB($w$) $\cup$ $\dec(x)$
                \EndFor

        \end{algorithmic}
    \end{algorithm}

\textbf{Search:} Algorithm \ref{ZSearch} shows the details of the search operation.
\begin{itemize}
    \item To search for all files that contain a keyword $w$ in the encrypted database EDB, a data user obtains the search token set $\{\mathcal{Q}^*_m\}_{m=1}^l$ by leveraging HKToken($\cdot$), where the time range is derived from the time of the search.
    Then, the server randomly chooses a matrix $mat$ from $Mat$ and calculates the trace between $mat$ and each token $\mathcal{Q}^*_m$ in the token set.
    If there is a result greater than 0, that means the keyword and timestamp match the query criteria and the result is the head block key $key$ of the searched keyword chain.
    \item Once obtaining the $key$, the server can proceed with the search operation of the chain structure.
    The server recovers the address of the encrypted identifier and the previous block key $kpr$.
    The server repeats this process by updating the block key until $key$ is $0^{\lambda}$.
    \item The operation has computation complexity of $O(a_w+|W|)$ and communication complexity of $O(n_w)$, because it employs matrix multiplication to retrieve the head block key and the address of the encrypted identifiers is not actually deleted.
\end{itemize}

\section{Security Analysis}\label{section_security}
In this section, we present the security analysis of our scheme.
We show that the HK technique is IND-CPA secure and our NIMS scheme is IND-CPA secure.

\begin{definition}
Considering a secure pseudo-random function $\pseudo$ and a cryptographic hash function $\hash$,
our scheme is $\mathcal{L}$\textit{-adaptive-secure} under the random oracle model. The collection of leakage functions $\mathcal{L}=(\mathcal{L}_{Setup}, \mathcal{L}_{Update}, \mathcal{L}_{Search})$ is written as follows: \\
\vspace{1 ex}\qquad $\mathcal{L}_{Setup}()=\varnothing,$\\
\vspace{1 ex}\qquad $\mathcal{L}_{Update}(Docs,op)=(op,\sum_{w\in W}|DB(w)|),$\\
\vspace{1 ex}\qquad $\mathcal{L}_{Search}(w)=(sp(w),$ TimeDB($w$), DelHist($w$)).
\end{definition}

\subsection{Security of HK technique}

\begin{definition}(Security of HK technique):
The semantic security of HK technique is defined via an IND-CPA game.
In this game, the adversary $\mathcal{A}$ can get the ciphertext of any message.
In particular, the game is described as follows:

\begin{itemize}
    \item \textbf{Setup}:
    The challenger $\mathcal{C}$ runs $Setup(1^{\lambda})$ to generate the master secret key $msk=\{ \mathcal{M}_1,  \mathcal{M}_2 \}$.

    \item \textbf{Phase 1}: $\mathcal{A}$ adaptively chooses several requests $P_{v_j}^{HK}$, for $j\in [q_1]$. On the $j$-th ciphertext request,  $\mathcal{A}$ submits an encryption query $v_j=(w, t, u)$ to the challenger $\mathcal{C}$.
    $\mathcal{C}$ responds with the ciphertext $\mathcal{P}_{v_j}^*$ via HKData$(v_j)$ (HKToken$(v_j)$).

    \item \textbf{Challenge}: On input messages $v_0$, $v_1$, $\mathcal{C}$ selects $b \in \{0,1\}$ and calculates the ciphertext $\mathcal{P}_{v_b}^*$ via HKData$(v_b)$ (HKToken$(v_b)$).

    \item \textbf{Phase 2}: $\mathcal{A}$ repeats the operations in \textbf{Phase 1}, and $\mathcal{C}$ responds with $\mathcal{P}_{v_j}^*$ for $q_1+1 \leq j \leq q_2$ as described above.

    \item \textbf{Guess}: $\mathcal{A}$ returns a guess $b'$ of $b$.

\end{itemize}

Our predicate-only HK construction is said to be IND-CPA secure
if for any PPT adversary $\mathcal{A}$, the advantage $Adv^{IND-CPA}_{HK,{\mathcal{A}}}(1^{\lambda})$ of $\mathcal{A}$ is negligible in $\lambda$, where
$$Adv^{IND-CPA}_{HK, {\mathcal{A}}}(\lambda)=\vert Pr[b'=b]- \frac{1}{2} \vert \leq negl(\lambda).$$
Here, $\lambda$ is the security parameter and $negl(\lambda)$ is a negligible function that takes $\lambda$ as a parameter.
\end{definition}

\begin{theorem}
HK technique is IND-CPA secure.
\end{theorem}

\begin{proof}
We give the proof of the indistinguishability of ciphertexts $\mathcal{P}^*_{v_0}$ and $\mathcal{P}^*_{v_1}$.
Given a query $v_b=(w,t,u)$, the keyword $w$ and the timestamp (time range) $t$ is transformed to an integer vector $\overrightarrow{A}=[a_1,\cdots,a_{n-1}]$ by leverage algorithms BWMA and Binary($\cdot$).
Suppose the combination vector $\overrightarrow{P}_{v_b}=[r\cdot (\overrightarrow{A}), u]$ is the index vector before encryption.
According to $GenLowTriMart(\cdot)$, randomly generate an $n \times n$-dimensions lower triangular matrix $\mathcal{P}_{v_b}$ where the main diagonal elements are filled with the index vector $\overrightarrow{P}_{v_b}$.
This matrix is then encrypted as $\mathcal{P}_{v_b}^*=\mathcal{M}_1 \times \mathcal{I}_x \times \mathcal{P}_{v_b} \times \mathcal{I}_x \times \mathcal{M}_2$.

According to the law of matrix multiplication, denote the product of $\mathcal{I}_x$, $\mathcal{P}_{v_b}$, and $\mathcal{I}_x$ as $\mathcal{D}$.
The element $d_{i,j}$ ($ i,j \in [1,n]$) of the matrix $\mathcal{D}$ is computed as
  \begin{equation}
    \begin{aligned}
      d_{i,j}= \sum_{k = 1}^{n}\sum_{m = 1}^{n}x_{i,k}v_{k,m}x_{m,j},
    \end{aligned}
  \end{equation}
  where $x$ and $v$ are the elements in $\mathcal{I}_x$ and $\mathcal{P}_{v_b}$,  respectively.
  We can observe that the element $d_{i,j}$ satisfies
  \begin{equation}
    \left\{
    \begin{array}{lc}
      d_{i,j}=0                       & 1 \leq i < j  \leq n , \\
      d_{i,j}=\overrightarrow{P}_{v_b}(i) & 1 \leq i = j  \leq n , \\
      d_{i,j}=*                       & otherwise.
    \end{array}
    \right.
  \end{equation}
  Here, $*$ denotes a fixed random value.
Since $r_u$ and the elements of the lower invertible matrix $\mathcal{I}_x$ are randomly selected by $\mathcal{C}$, the matrix $\mathcal{D}$ is different even if the same message is selected.

  Then, we compute $\mathcal{P}^*_{v_b}=\mathcal{M}_1 \times \mathcal{D} \times \mathcal{M}_2$.
  The element $p_{i,j}$ ($ i,j \in [1,n]$) of the matrix $\mathcal{P}_{v_b}$ is computed as
  \begin{equation}
    \begin{aligned}
      p_{i,j}= \sum_{k = 1}^{n}\sum_{o = 1}^{n}m^{(1)}_{i,k}d_{k,o}m^{(2)}_{o,j},
    \end{aligned}
  \end{equation}

  Without loss of generality,  we can observe two elements of the matrix $\mathcal{P}^*_{v_b}$,  i.e.,  $p_{1,1}$ and $p_{n,n}$, where $$p_{1,1}= \sum_{k = 1}^{n}\sum_{o = 1}^{n}m^{(1)}_{1,k}d_{k,o}m^{(2)}_{o,1},$$
  $$p_{n,n}= \sum_{k = 1}^{n}\sum_{o = 1}^{n}m^{(1)}_{n,k}d_{k,o}m^{(2)}_{o,n}.$$
  We note that $p_{1,1}$ is composed of the multiplication of each element $d_{k,o}$ and two fixed constants $m^{(1)}_{1,k}$, $m^{(1)}_{o,1}$.
  $p_{n,n}$ is the same as $p_{1,1}$.
  Even if the elements of vector $\mathcal{M}_1$ and $\mathcal{M}_2$ are changeless, we can obtain different values as long as one element of the matrix $\mathcal{D}$ is different.
  In summary, the ciphertext $\mathcal{P}^*_{v_b}$ is not the same even if the same message $(w,t,u)$ is chosen.

  In Phase 1 and Phase 2 of the security game,  $\mathcal{A}$ selects different message $v_j=(w,t,u)$ each time and obtains the corresponding ciphertext $\mathcal{P}^*_{v_j}$, $j \leq q_2$.
  According to our analysis, $\mathcal{P}^*_{v_j}$ is a random matrix, and the ciphertexts appear random to $\mathcal{A}$. As a result, the adversary $\mathcal{A}$ cannot distinguish which message is encrypted,  given a ciphertext encrypted by using the message selected by $\mathcal{A}$. Hence,  we have $$Adv^{IND-CPA}_{\Sigma,{\mathcal{A}}}(\lambda)=\vert Pr[b'=b] - \frac{1}{2} \vert \leq negl(\lambda).$$
\end{proof}

\subsection{Security Analysis of NIMS}
We prove that  the security of NIMS is IND-CPA secure in the random oracle model, under the security of the underlying primitives.
\begin{theorem}
Assuming $\pseudo_1$ and $\pseudo_2$ are secure pseudo-random functions, $\hash$ is a cryptographic hash function, and HK technique is IND-CPA secure, NIMS is $\mathcal{L}$\textit{-adaptive-secure} in the random oracle model.
\end{theorem}

\begin{proof}
We prove the security of NIMS by defining a series of games.
In the following, we denote the i-th game by $Game_i$ and $Pr_{G_i}[E]$ denotes the probability that an event E occurs in $Game_i$.

$Game_0$: The game is the same as the real-world DSSE security game $Real$.
$$Pr[Real_{\mathcal{A}}^{NIMS}(\lambda)=1]=Pr[G_0=1]$$

$Game_1$: Different from the previous game, we replace the calls to the pseudo-random function $\pseudo_1$ and $\pseudo_2$ by picking random elements of the appropriate range.
We maintain a table \textit{Key} to store $(ctr,w, key_w)$ pairs and a table $\textit{ID}$ to store $(ind,eid)$ pairs.
In the update/search protocol, when operating on $(ctr,w)$, the experiment first checks whether there has existed $(ctr,w,key_w)$ in \textit{Key}.
If there is an entry in the table, the corresponding value can be retrieved directly from it,  otherwise, a random output $key_w$ replaces the call to the pseudo-random function and it is stored in a table.
The search on $ind$ is as same as $(ctr,w)$.
From the adversary's perspective, the advantage in distinguishing between $Game_1$ and $Game_0$ is equal to that between a pseudo-random function $\pseudo$ and a truly random function.
We conduct a reduction ${B}_1$ to distinguish between $\pseudo$ and a truly random function, such that
$$Pr[G_1 =1] - Pr[G_0 = 1] \leq 2 Adv_{\pseudo,\mathcal{B}_{1}}^{PRF}(\lambda)$$

$Game_2$: The difference with Game2 is that the hash function $\hash_1$ in the update protocol is replaced by random picking strings.
We maintain a tables $H_1$ to record ($key||0$, $addr$) pairs and ($key||1$, $mask$) pairs.
These strings are used to answer the random oracle in the search query.
When the same input is recalled to $\hash$ and $\hash_1$, the experiment retrieves the table and returns the result directly.
The behavior of $Game_2$ and $Game_1$ is identical to the adversary, except that some probabilistic inconsistencies can be observed in $Game_2$.
There exist a  reduction $\mathcal{B}_{2}$ such that
$$Pr[G_2 =1] - Pr[G_1 = 1] \leq 2 Adv_{H,\mathcal{B}_{2}}^{hash}(\lambda)$$

$Game_3$: In this game, the index matrix and trapdoor matrices are generated by running the simulator $\mathcal{S}_{HK}$ of HK technique.
The index matrix $\mathcal{I}^*$ is generated by random choosing $n\times n$-dimensions matrix.
The trapdoor matrices are generated the same as the index matrix.
By the description of $Game_3$ and $Game_2$, we know that the real game of HK technique with $\mathcal{B}_3$ perfectly simulates $Game_2$, so we have that
$$Pr[G_3 =1] - Pr[G_2 = 1] \leq  Adv_{\mathcal{B}_{3}}^{HK}(\lambda)$$

$Simulator$: A view is generated only given the leakage function by the simulator, where the value of LSC structure and HK technique are randomly picked.
To avoid explicitly using the keyword $w$, we replace $w$ with $min$ $sp(w)$, which denotes the first index at which $w$ emerges in the search pattern.
Hence, we have
$$Pr[G_3 =1] = Pr[Ideal_{\mathcal{A}}^{NIMS} = 1] $$

$Conclusion$: By combining all games, there exists three adversaries $\mathcal{B}_1$, $\mathcal{B}_2$, and $\mathcal{B}_3$ such that

 \begin{equation}
    \begin{aligned}
     &Pr[Real_{\mathcal{A}}^{NIMS}(\lambda) =1]
      = Pr[Ideal_{\mathcal{A},\mathcal{S},\mathcal{L}}^{NIMS}(\lambda) = 1]
     \\
     & \leq 2Adv_{F,\mathcal{B}_{1}}^{PRF}(\lambda) +  2Adv_{H,\mathcal{B}_{2}}^{hash}(\lambda)
     +Adv_{\mathcal{B}_{3}}^{HK}(\lambda).
     \end{aligned}
\end{equation}

\end {proof}

\section{Performance Evaluation}\label{section_experiment}
In this section,  we evaluate our scheme and compare it with five related schemes: Janus++\cite{10.1145/3243734.3243782}, CLOSE-FB\cite{9237959}, FAST\cite{8329523}, FASTIO\cite{8329523}, and Bestie\cite{10.1007/978-3-030-88428-4_1}. All schemes achieve forward privacy while NIMS, CLOS-FB, and Bestie also achieve backward privacy.

\subsection{Implementation and Dataset}

We implemented our experiment in Python 3 and used pycrypto to implement cryptographic primitives.
We used the pseudo-random function $F_1$, $F_2$ with HMAC-SHA-256.
The hash function $H_1$ and $H_2$ were instantiated using HMAC-SHA-512.
The SE was instantiated using AES.
We set the maximum number $d=5$ of tags to be punctured in Janus++ and use the maximum length  \textit{CLen} = 1200 of Fish-Bone Chain in CLOSE-FB.
Our experiment was run on the Intel Core i7 CPU system at 2.6 GHz and 16 GB RAM.
We configure FAST, FASTIO, CLOSE-FB and Janus++ for single client setup.
To demonstrate our scheme is multi-client non-interactive, we developed  NIMS with three separate data owner, data user, and server processes.

\begin{table}[htbp]
  \caption{Database size.}
  \centering
  \begin{tabular}{cccc}
    \hline
    \specialrule{0em}{1pt}{1pt}
     & $DB_1$ & $DB_2$  \\ \hline
    \specialrule{0em}{1pt}{1pt}
    Keywords   & 5,000      & 13,475      \\
    \specialrule{0em}{1pt}{1pt}
    Keyword/document pairs & 6,032,672     &11,451,557      \\
    \specialrule{0em}{1pt}{1pt}
    Documents  & 12,000   & 23,000   \\ \hline
  \end{tabular}
  \label{table:database}
\end{table}

In the experiments, we chose a fraction of the Enron Email Dataset as the dataset.
First, we employ RAKE to extract the keywords from the dataset, and then randomly select 23000 documents and extracted 13,475 keywords to construct keyword/document pairs. Finally, we obtain 11,451,557 keyword/document data pairs in total.
Meanwhile, we also built a database based on the Enron email dataset to investigate the effects of database size.
In Table \ref{table:database}, all of the databases that we used in our research are described in detail.
In our evaluation, the results of each experiment are averaged from 10 tests.
To demonstrate the performance of our proposed scheme, we have performed a comprehensive evaluation in terms of computation time and communication overhead.

\subsection{Evaluation of Addition}

\begin{table}[htbp]
  \caption{Comparison of client-side storage cost.}
  \centering
  \begin{tabular}{cccc}
    \hline
    \specialrule{0em}{1pt}{1pt}
    Scheme & $DB_1$ & $DB_2$ \\ \hline
    \specialrule{0em}{1pt}{1pt}
    Janus++  &34,253KB   & 72,179KB \\
    \specialrule{0em}{1pt}{1pt}
    Bestie  & 29,242KB      & 63,158KB      \\
   \specialrule{0em}{1pt}{1pt}
    FASTIO   & 31,051KB      & 66,617KB      \\
    \specialrule{0em}{1pt}{1pt}
    CLOSE-FB   & 1KB     & 1KB     \\
    \specialrule{0em}{1pt}{1pt}
    NIMS  & 18KB   & 40KB   \\ \hline
  \end{tabular}
  \label{table: Storage cost}
\end{table}

Table \ref{table: Storage cost} displays the client storage comparison of NIMS with other schemes before any deletions.
When implementing FAST, FASTIO, Bestie, and Janus++, we noticed that these schemes have to provide all keywords of the deleted document to generate update tokens.
Therefore,  an inverted index (or a forward index) is necessary to be stored on the client side.
In contrast, our scheme only stores a global variable and the keyword set.
When performing the deletion, our scheme only uploads the address of the encrypted identifier.
In FAST, FASTIO, and Bestie, the client has to hold a map for storing states except for the inverted index.
In Janus++, the client maintains two dictionaries for addition and deletion and a local key share for puncturable encryption.
From Table \ref{table: Storage cost}, it can be observed that the storage cost of FAST, FASTO, Bestie, and Janus++ are enlarged as the size of the database increases. On the contrary, the storage cost of NIMS and CLOSE-FB are small for client storage.
According to Table \ref{table: Storage cost}, NIMS and CLOSE-FB have small client storage costs while FAST, FASTO, Bestie, and Janus++ have storage costs that increase with database size.

\subsection{Evaluation of Addition}
\begin{figure}[htbp]
  \centering
  \includegraphics[width=0.45\textwidth]{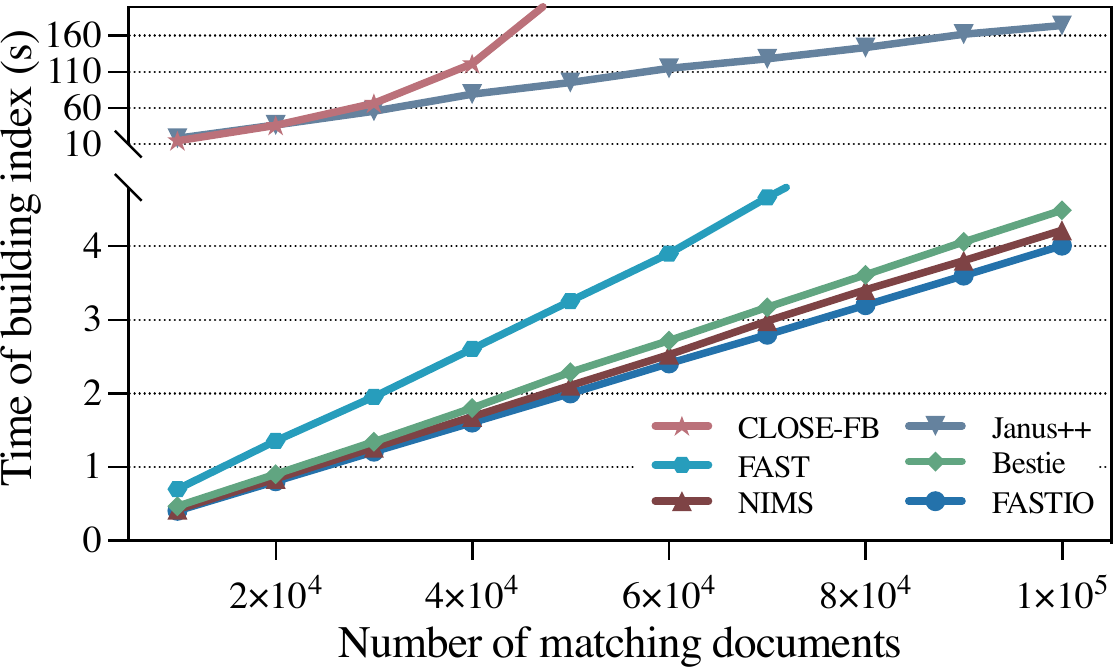}
  \caption{Comparison of different schemes in the encryption process.}
  \label{fig:index generation}
\end{figure}
To validate the addition performance, we initialize $10^5$ key-value pairs and add them to the database.
Fig.\ref{fig:index generation} compares the computational overhead of adding files for various schemes.
In NIMS, pseudo-random functions and hash functions are employed to construct the LSC structure, and then HL technique is applied to encrypt the head block keys of all keyword chains.
In CLOSE-FB,  while the global counter \textit{CLen} is limited, we only can perform the 'add' operation with the constant \textit{CLen} times.
When \textit{CLen} is exhausted, we search all the keywords to recover the document set \textit{DB} and upload ciphertexts to the server after re-encrypting \textit{DB}.
When the database is large, \textit{CLen} is soon exhausted and re-encryption costs much time in the addition operation.
In Janus++, for each keyword/document pair, Janus++ needs to calculate a GGM tree with $d$ layers, each encrypted with AES.
The upload efficiency increases as the parameter $d$ (i.e. the maximum number of documents to be deleted) increases.
In the comparison experiment, to improve the update efficiency in Janus++, we set the parameter $d=5$, where at most 5 files can be deleted.
In FAST, AES encryption was used to implement pseudorandom permutation and AES decryption to the inverse permutation.
In FASTIO and Bestie, only the pseudorandom function and hash function are employed.
Although NIMS has a little higher update time cost than Bestie and FASTIO,  NIMS is more useful since both Bestie and FASTIO should share the latest keyword status (i.e. keyword counter) to the data user in the case of multi-client.

\subsection{Evaluation of Deletion}

\begin{table}[htbp]
  \caption{Deletion communication cost}
  \centering
  \begin{tabular}{cccc}
    \hline
    \specialrule{0em}{1pt}{1pt}
     & 1,000 & 5,000  &10,000  \\ \hline
    \specialrule{0em}{1pt}{1pt}
    NIMS   & 9 B     & 9 B    &9 B  \\
    \specialrule{0em}{1pt}{1pt}
    Bestie & 67.39 KB     &336.92 KB   &673.84 KB  \\
    \specialrule{0em}{1pt}{1pt}
    FAST, FASTIO & 58.60 KB   & 292 KB  &585.95 KB   \\
    \specialrule{0em}{1pt}{1pt}
    CLOSE-FB  &59.58 KB &297.86 KB &595.71 KB\\ \hline
  \end{tabular}
  \label{table:deletion}
\end{table}

All compared schemes that perform deletion operations have to generate search tokens based on ($ind, w$) pairs, which means numerous deletion tokens are sent to the server when a document is to be deleted.
What makes our solution different from others is that we do not have to store all keywords of the deleted document.
Our proposed scheme could generate a token based on the identifier to delete the document.

We take deletion operations into account and test deletion performance.
This experiment randomly selects three groups of documents with the different number of documents to be deleted for testing NIMS, Bestie, CLOSE-FB, FAST and FASTIO deletion performance.
Janus++ is not considered in the comparison of deletion bandwidth cost because of the high addition time when we set the maximum deletion number $d$ to be larger than 100. 
Because the deletion operations of FAST and FASTIO are similar, we have combined them into one group.
Table \ref{table:deletion} shows the deletion communication cost based on the different number of keywords in the deleted document.
NIMS achieves the lowest bandwidth cost compared with other schemes.

\subsection{Evaluation of Search}

\begin{figure}[htbp]
  \centering
  \includegraphics[width=0.45\textwidth]{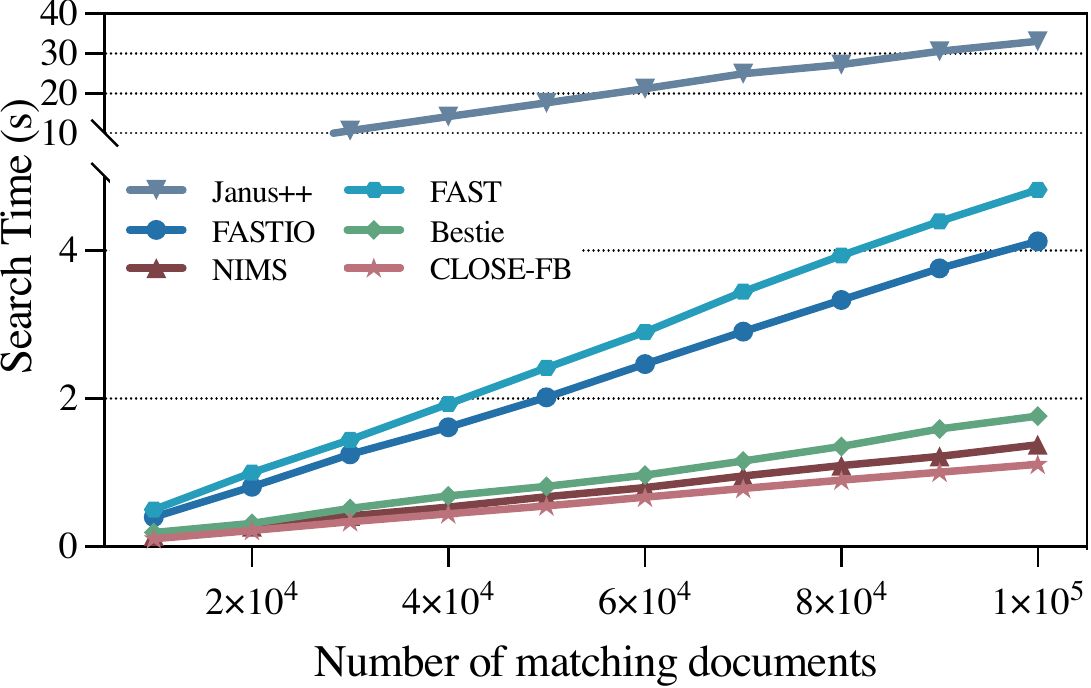}
  \caption{Comparison of different schemes in the search process.}
  \label{fig:search time}
\end{figure}

Fig.\ref{fig:search time} compares the search performance of six schemes on the same database.
As the number of matched documents increases with database size, the results show that the search time is related to database size.
Since Janus++ uses AES and GGM trees to achieve puncturable encryption, each plaintext recovered requires the punctured key.
For FAST and FASTIO, use inverse permutation to support forward privacy, but it also achieves better search efficiency.
The search efficiency of our scheme is slightly lower than CLOSE-FB, since NIMS needs to implement a matching operation to achieve non-interactive i.e., the server searches the header block keys first.
Obviously, the search time of our proposed scheme slowly increases with the number of matching entries.

\begin{figure}[H]
  \centering
  \includegraphics[width=0.45\textwidth]{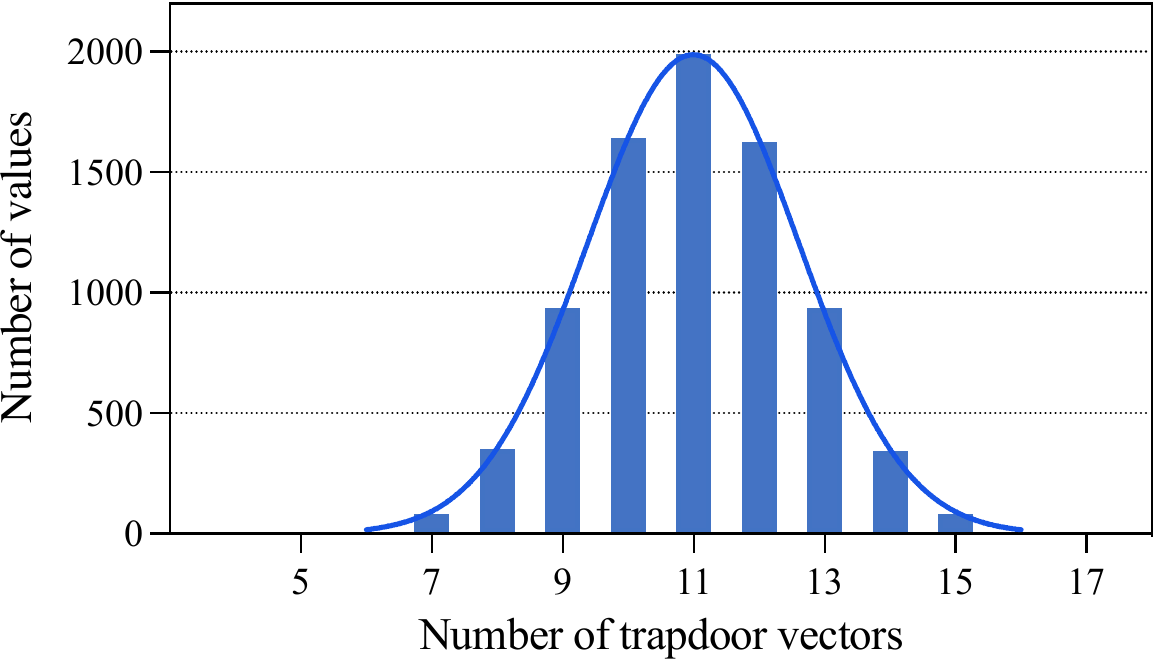}
  \caption{Frequency distribution of trapdoor vector numbers.}
  \label{fig:Frequency distribution}
\end{figure}

We make frequency statistics on the number of trapdoor vectors generated in 100,000 seconds, described in Fig.\ref{fig:Frequency distribution}
We observe a Gaussian-like frequency distribution, where the number of trapdoor vectors is mostly ranged from 8 to 12.
To obtain the head block key, the server randomly selects encrypted links and a trapdoor vector set to perform matrix multiplication.
Assuming that the time to calculate the matrix trace is $T$, the time to acquire the head block key is at most $11\times T\times|W|$.

\begin{figure}[H]
  \centering
  \includegraphics[width=0.45\textwidth]{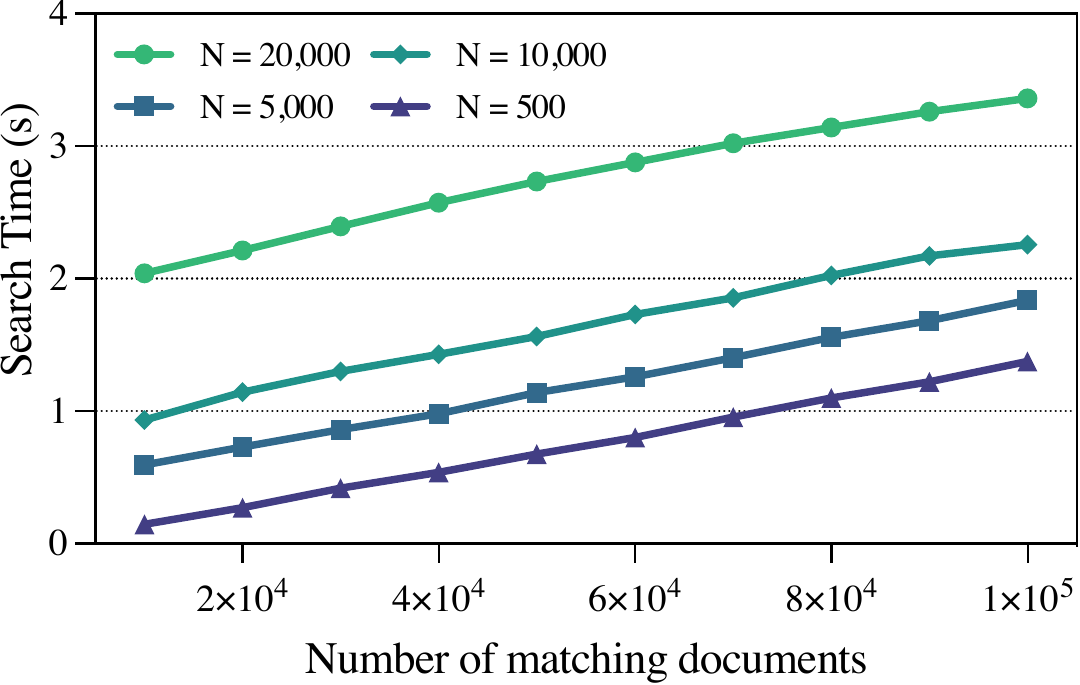}
  \caption{Search time with various numbers of keywords in the search process.}
  \label{fig:search time keywords}
\end{figure}

We measured the search latency under various parameter settings to explore more about the search performance of NIMS.
We expect to understand the effect of $|W|$ on searching the head block key time.
Fig.\ref{fig:search time keywords} describes the maximum search time with different numbers of keywords in NIMS. Noted that the server most probably does not need to compute traces of all keyword state matrices and search matrices.
With the growth of the number of keywords, the time for fetching the same number of matched documents is relatively increased, because the number of encrypted keyword matrices is related to the number of keywords.
Matrix multiplication takes a small amount of time, whereas 100,000 multiplications only take 1 second.
The increased rate of search time with the different number of keywords is similar.
We use HK technique to efficiently realize non-interactive forward privacy.

\subsection{Summary}
Our testing schemes are all DSSE schemes with forward privacy while NIMS, CLOSE-FB, Janus++, and Besite achieve backward privacy.
The performance of addition in FASTIO is similar to our scheme.
However, FASTIO performs a little inferior on the search operation because of the inverse permutation using AES decryption.
The search performance of CLOSE-FB is the best, but \textit{CLen} limits the update performance.
On balance, NIMS achieves optimal update and search performance.

\section{Related Work}
SSE is a primary apporach for encrypted data searching in cloud servers~\cite{song2000practical}.
In the proposed SSE scheme, the client outsources the encrypted documents to an untrusted server and queries for documents that contain a specific keyword with linear search time.
Golle et al.~\cite{golle2004secure} proposed a scheme to generate encrypted indices for keywords based on Bloom filters, but the efficiency is not high.
Subsequent research in \cite{curtmola2011searchable, chase2010structured, wang2015inverted} mainly focuses on static datasets.

DSSE was proposed to support data updates, allowing files to be added/deleted in the database.
Kamara et al.~\cite{kamara2012dynamic} constructed a DSSE scheme that achieved sub-linear search but updating documents would leak the document structure.
Since then, several schemes \cite{kamara2013parallel, naveed2014dynamic, cash2014dynamic} realizes dynamic searchable encryption.
However, malicious servers can easily obtain additional information by observing repeated queries and other information.
As file-injection and leakage-abuse attacks have been demonstrated,  researchers realized that forward privacy is an important research direction of DSSE.
Forward privacy focuses on keeping update operations from linking to previous search queries.
Schemes \cite{bost2016ovarphiovarsigma, kim2017forward,yang2017rspp} were proposed that supports forward privacy.

Just recently, Bost et al.~\cite{bost2017forward} formally defined the notion of backward privacy.
Backward privacy prevents the server from knowing the deleted document.
Since then, forward and backward privacy became the basic requirements of DSSE.
Some schemes with forward and backward privacy \cite{bost2017forward,ghareh2018new,sun2021practical,9237959} were proposed.
But they all focus on the single-client framework which is not practical in the real-world case.

Although there are many schemes \cite{jarecki2013outsourced, faber2015rich, zuo2019dynamic} with the multi-client settings, they need frequent interactions between data owner and users since keyword states are recorded on the data owner side.
Some schemes introduce a proxy server to share secret keys and keyword states, such as the scheme proposed by Wang et al.~\cite{10.1007/978-3-030-02744-5_9}.
In addition, since existing single-client solutions keep the keyword status locally \cite{ghareh2018new,9237959}, a simple extension of a single solution to a multi-client setup would require frequent interactions between data owners and data users.
To save communication overhead and reduce the risks from interaction, it is necessary to avoid frequent interactions. 
Sun et al. constructed two non-interactive multi-client SSE protocols \cite{10.1007/978-3-319-45744-4_8, sun2020non}.
However, the two schemes only focus on the static database.
After that, \cite{sun2018dynamic} was proposed which achieves multi-client non-interactive.
But, it can not achieve forward and backward privacy and may suffer from file injection attacks.

\section{CONCLUSION}
In this work, we propose the first non-interactive multi-client DSSE scheme with small client storage.
Observing that most existing schemes store an inverted index or a forward index incurring a heavy computation cost, we set a global variable that binds all keywords to signal status.
Meanwhile, we hold a table storing the encrypted file identifiers to achieve efficient deletion operation.
We further propose a more practical chain structure that implicitly links all keyword/document pairs and takes up small client-side storage.
To achieve multi-client non-interactive, we design a hidden key technique 
to encrypt the head block key of the keyword chains.
The technique leverages the time range query to ensure forward privacy.
Furthermore, we propose a new framework for constructing a non-interactive multi-client DSSE scheme with small client storage.
Moreover, we implemented our scheme, CLOSE-FB, FAST, FASTIO, and Janus++ for evaluation. Experimental results demonstrate our scheme's superior efficiency and dependability.


%

\ifCLASSOPTIONcaptionsoff
  \newpage
\fi

\bibliographystyle{IEEEtran}
\bibliography{Reference}



\begin{IEEEbiography}
 [{\includegraphics[width=1in,height=1.25in,clip,keepaspectratio]{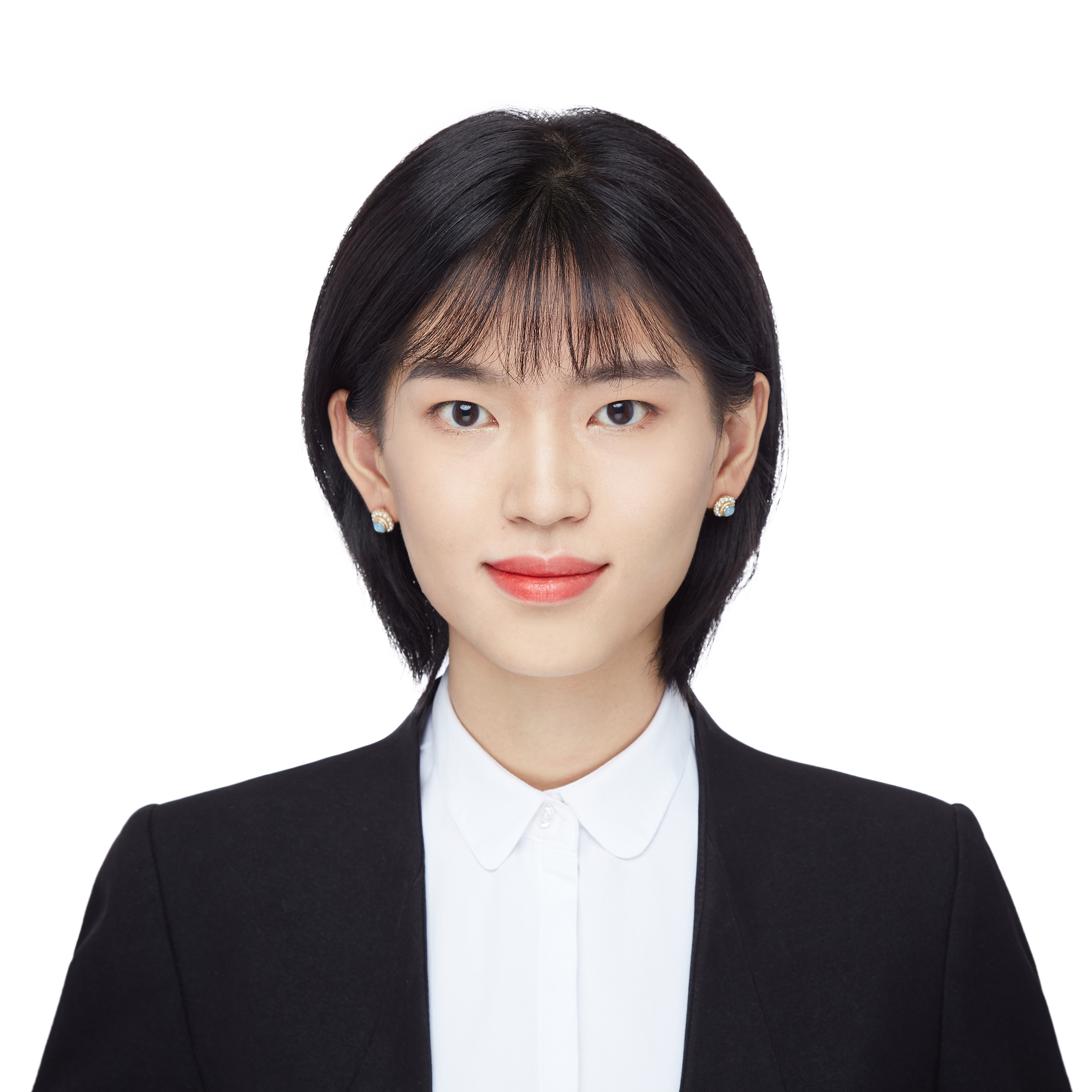}}]{Hanqi Zhang} received the B.S. degree in computer science from Beijing Institute of Technology, Beijing, in 2018. She is currently pursuing the Ph.D. degree at the School of Cyberspace Science and Technology, Beijing Institute of Technology. Her current research interests include security and privacy in searchable symmetric encryption.
\end{IEEEbiography}

\begin{IEEEbiography}
 [{\includegraphics[width=1in,height=1.25in,clip,keepaspectratio]{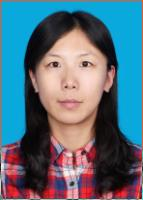}}]{Chang Xu} received her Ph.D. degree in computer science from Beihang University, Beijing, China, in 2013. She is currently an associate professor at the School of Cyberspace Science and Technology, Beijing Institute of Technology. Her research interests include security \& privacy in VANET, and big data security.
\end{IEEEbiography}

\begin{IEEEbiography}
 [{\includegraphics[width=1in,height=1.25in,clip,keepaspectratio]{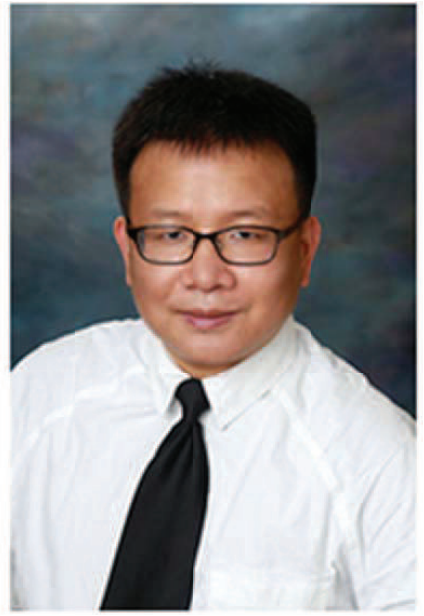}}]{Rongxing Lu} [S'09-M'11-SM'15-F'21] is an associate professor at the Faculty of Computer Science (FCS), University of New Brunswick (UNB), Canada. Dr. Lu is an IEEE Fellow. His research interests include applied cryptography, privacy enhancing technologies, and IoT-Big Data security and privacy. He has published extensively in his areas of expertise (with H-index 72 from Google Scholar as of November 2020), and was the recipient of 9 best (student) paper awards from some reputable journals and conferences. Currently, Dr. Lu serves as the Vice-Chair (Conferences) of IEEE ComSoc CIS-TC (Communications and Information Security Technical Committee).
\end{IEEEbiography}

\begin{IEEEbiography}
 [{\includegraphics[width=1in,height=1.25in,clip,keepaspectratio]{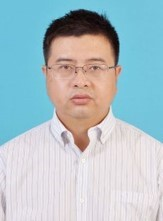}}]{Liehuang Zhu} received his Ph.D. degree in computer science from Beijing Institute of Technology, Beijing, China, in 2004. He is currently a professor at the School of Cyberspace Science and Technology, Beijing Institute of Technology. His research interests include security protocol analysis and design, group key exchange protocols, wireless sensor networks, cloud computing, and blockchain applications.
\end{IEEEbiography}

\begin{IEEEbiography}
 [{\includegraphics[width=1in,height=1.25in,clip,keepaspectratio]{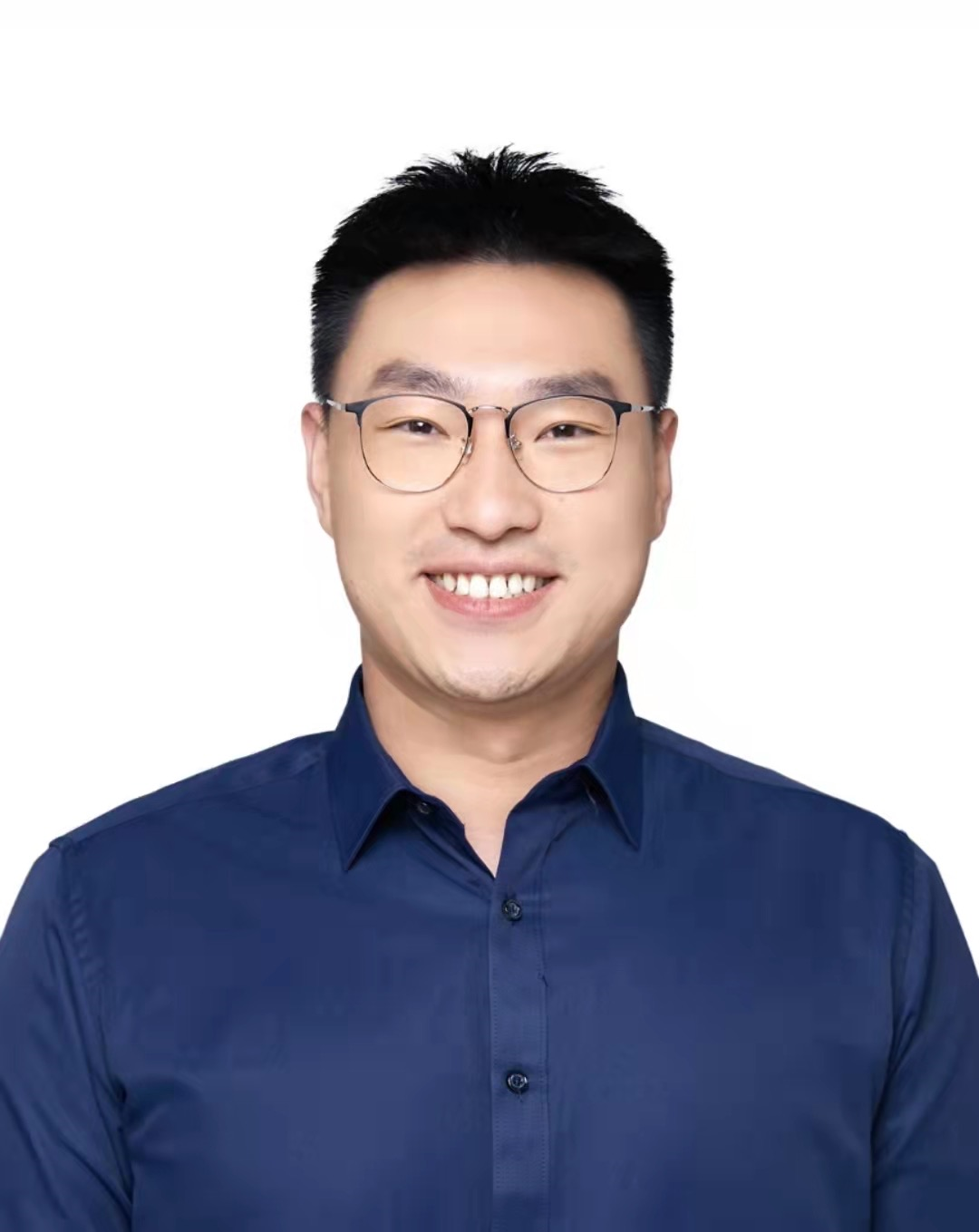}}]{Chuan Zhang} received his Ph.D. degree in computer science from Beijing Institute of Technology, Beijing, China, in 2021.  From Sept. 2019 to Sept. 2020, he worked as a visiting Ph.D. student with the BBCR Group, Department of Electrical and Computer Engineering, University of Waterloo, Canada. He is currently an assistant professor at School of Cyberspace Science and Technology, Beijing Institute of Technology. His research interests include secure data services in cloud computing, applied cryptography, machine learning, and blockchain.
\end{IEEEbiography}

\begin{IEEEbiography}
 [{\includegraphics[width=1in,height=1.25in,clip,keepaspectratio]{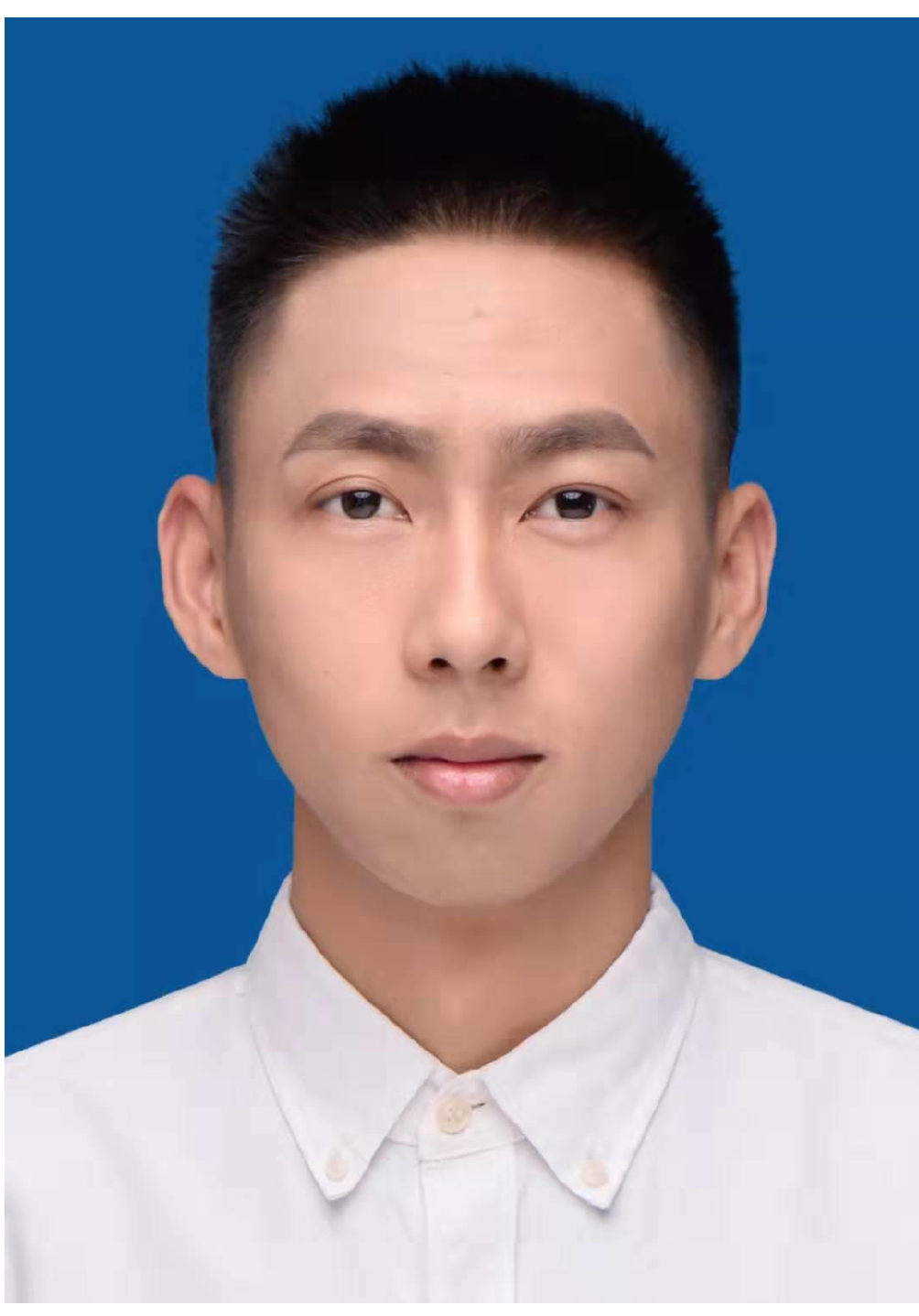}}]{Yunguo Guan} is a PhD student of  the Faculty of Computer Science, University of New Brunswick, Canada. His research interests include applied cryptography and game theory.
\end{IEEEbiography}




\end{document}